\documentclass[a4paper,UKenglish,cleveref,autoref,thm-restate]{lipics-v2021}

\title{$\VP$, $\VNP$ and Algebraic Branching Programs over Min-Plus Semirings}

\author{Balagopal Komarath}{Department of Computer Science and Engineering, IIT Gandhinagar, Gujarat, India}{bkomarath@rbgo.in}{}{}

\author{Harshil Mittal}{Department of Computer Science and Engineering, IIT Madras, Chennai, India}{harshil@cse.iitm.ac.in}{https://orcid.org/0000-0002-4819-5711}{}

\author{Jayalal Sarma}{Department of Computer Science and Engineering, IIT Madras, Chennai, India}{jayalal@cse.iitm.ac.in}{https://orcid.org/0000-0002-4819-5711}{}

\authorrunning{Komarath et al.} 

\Copyright{Balagopal Komarath and Harshil Mittal and Jayalal Sarma} 



\category{} 

\relatedversion{} 

\nolinenumbers 

\EventEditors{John Q. Open and Joan R. Access}
\EventNoEds{2}
\EventLongTitle{42nd Conference on Very Important Topics (CVIT 2016)}
\EventShortTitle{CVIT 2016}
\EventAcronym{CVIT}
\EventYear{2025}
\EventDate{December 24--27, 2016}
\EventLocation{Little Whinging, United Kingdom}
\EventLogo{}
\SeriesVolume{42}
\ArticleNo{23}

\usepackage{graphicx} 
\usepackage{nicematrix,complexity}
\usepackage{amsmath}
\usepackage{amsthm}
\usepackage{amssymb}
\usepackage{yhmath}

\usepackage[T1]{fontenc}
\usepackage{babel}

\usepackage[labelsep=space]{caption}

\newtheorem{thm}{Theorem}
\newtheorem{lem}[thm]{Lemma}

\usepackage{collect}

\def\movetoappendix{1}
\definecollection{appendix}
\makeatletter
\newenvironment{aproof}[2]
  { \@nameuse{collect}{appendix}
  { \subsection{#1} \label{#2} \begin{proof} } {\end{proof}}
  }{\@nameuse{endcollect}}
\makeatother

\makeatletter
\newenvironment{appsection}[2]
  { \@nameuse{collect}{appendix}
  { \subsection{#1} \label{#2} }
  {}
  }{\@nameuse{endcollect}}
\makeatother

\ifthenelse{\equal{\movetoappendix}{0}}{
        
        \renewenvironment{appsection}[2]{} {}
}{}

\begin{document}
\maketitle
\begin{abstract}
Arithmetic circuit complexity studies the complexity of computing polynomials using only arithmetic operations such as addition, multiplication, subtraction, and division. Polynomials over rings of integers model counting problems. Similarly, polynomials over semirings such as tropical semirings model optimization problems. Circuits over semirings then model so called \emph{pure} algorithms, algorithms that only use the operations in the semiring.

In this paper, we do a complexity-theoretic study of the power and limitations of circuits (which represent dynamic programs) over semirings:
\begin{itemize}
    \item We define $\mathsf{VNP}$ over min-plus semirings, which can faithfully represent problems such as computing min-weight perfect matchings and min-weight Hamiltonian cycles where we have efficiently verifiable certificates. Unlike over rings, we complement the values in the certificate for free as complementation is impossible over min-plus semirings. We prove a dichotomy theorem that states that if we only complement logarithmically many values, this class is same as $\mathsf{VP}$ over min-plus semirings. If we complement super-logarithmically many values, then $\mathsf{VNP} \neq \mathsf{VP}$.
    \item We consider constant-width ABPs (which are also called incremental dynamic programs that are restricted to use only a constant number of registers) and show that even simple problems like computing the min-weight $2$-edge-matching is impossible with width $2$ (or $2$ registers). However, with width $3$ (or $3$ registers), such programs can compute everything. More generally, we show that constant-depth formulas are efficiently simulated by constant-width ABPs.
    \item We show that an exponential hypercube sum (min in the semiring) over even provably weak models such as width-$2$ ABPs and products of linear forms are the same as $\mathsf{VNP}$.
\end{itemize}


\acknowledgements{
We thank Dhara Thakkar for helpful discussions during the initial phase of this work. The first author would like to thank Bireswar Das and Anuj Tawari for helpful discussions related to tropical circuits. The second author is also grateful to the anonymous reviewers of his PhD thesis \cite{harshil2025diverse} (which contained a preliminary version of some part of the work presented herein, and it was done while the second author was a PhD student at IIT Gandhinagar).}

\end{abstract}

\newpage
\section{Introduction}

Arithmetic circuit complexity studies the complexity of computing polynomials using only addition and multiplication over an underlying field $\mathbb{F}$. The central open problem in arithmetic circuit complexity is an analogue of the $\mathsf{P}$ vs $\mathsf{NP}$ problem called the $\mathsf{VP}_{\mathbb{F}}$ vs $\mathsf{VNP}_{\mathbb{F}}$ problem. The class $\mathsf{VP}_{\mathbb{F}}$ consists of polynomial-families\footnote{More precisely, polynomial-families with polynomially-bounded degree \& number of variables.} over any field $\mathbb{F}$ that can be computed by polynomial-sized arithmetic circuits (using operations $+$ and $\times$), and the class $\mathsf{VNP}_{\mathbb{F}}$ consists of families expressible as an exponential size hyper-cube sum over families in $\mathsf{VP_{\mathbb{F}}}$. Here, the exponential size hyper-cube sum is the analogue of the $\exists$ quantifier over exponentially many certificates that are verifiable in $\mathsf{P}$ that defines $\mathsf{NP}$. It is also known that assuming the generalized Riemann Hypothesis, over the field of complex numbers $\mathbb{C}$, separating $\textsf{VP}_{\mathbb{C}}$ from  $\textsf{VNP}_{\mathbb{C}}$ is necessary to separate $\P$ from $\NP$ as well.

A natural generalization of arithmetic circuits is to allow semi-rings instead of rings. Such circuits computing over Boolean semi-rings (see \cite{Vollmer-Survey}) and $(\min, +)$-semi-rings over the natural numbers, which model an interesting subset of algorithms for optimization problems, (see Jukna~\cite{jukna2015lower}) have been studied. For example, several graph minimization problems take as input an $n$-vertex graph $G$ with costs (often reals or non-negative reals, i.e., from $\mathbb{R}$ or $\mathbb{R}_{\geq 0}$) associated to its edges (we may assume that missing edges of $G$ are present with a cost of $\infty$), and the goal is to minimize a certain linear function of costs of edges in $H$ over all subgraphs $H$ of $G$ satisfying a certain property $\Pi$. That is, such a problem aims to compute:

$$\underset{\substack{H:~H \text{ is a subgraph}\\ \text{ of } G \text{ satisfying } \Pi}}{\min}\Bigg(\underset{e\in E(H)}{\sum}c_{e}\cdot \operatorname{cost}(e)+d_e\Bigg),$$

where $c_e$'s and $d_e$'s are constants (i.e., independent of input edge costs) from $\mathbb{R}$ or $\mathbb{R}_{\geq 0}$. In particular, for many natural optimization problems, such as Shortest $s$-$t$ Path, Minimum Perfect Matching, Minimum Spanning Tree and Minimum Hamiltonian Cycle problems, the property $\Pi$ is that the subgraph $H$ of $G$ is an $s$-$t$ path, a perfect matching, a spanning tree and a Hamiltonian cycle of $G$ respectively, all $c_{e}$'s are $1$ and all $d_e$'s are $0$. When all $c_e$'s are non-negative integers (i.e., from $\mathbb{N}$), we can model the problem equivalently as the computation of the following polynomial function over min-plus semiring $\mathsf{R}:=(\mathbb{R}\cup\{\infty\}, \oplus, \otimes)$ or $\mathsf{R}^{+}:=(\mathbb{R}_{\geq 0}\cup\{\infty\}, \oplus, \otimes)$, where $\oplus$ (i.e., semiring's addition) denotes the minimum operation and $\otimes$ (i.e., semiring's multiplication) denotes the usual addition, and the value substituted for variable $x_e$ is the input $\operatorname{cost}(e)$ assigned to edge $e\in E(G)$:
$$\underset{\substack{H:~H \text{ is a subgraph}\\ \text{   of } G \text{ satisfying } \Pi}}{\bigoplus}\Bigg(\underset{e\in E(H)}{\bigotimes}d_e\otimes \underbrace{x_e\otimes x_e \otimes \ldots \otimes x_e}_{c_e \text{ times}}\Bigg)$$

Therefore, it is interesting to study the complexity of computing these polynomials using models that use only $\min$ and $+$ operations.

\subsection{Pure \texorpdfstring{$(\operatorname{min},+)$}{(min, +)} DP algorithms as Circuits}

Many problems of the above form admit polynomial-time dynamic programming (DP) algorithms. For example, Shortest $s$-$t$ Path problem can be solved in $\mathcal{O}(n^3)$ time using Bellman-Ford algorithm \cite{bellman1958routing}. When edge costs are non-negative, this algorithm can be formulated as a DP. The DP table has entries $\mathbb{T}(v, \ell)$'s for all $v\in V(G), \ell\in [n]$, where $\mathbb{T}(v, \ell)$ is the minimum cost of any $s$-to-$v$ path of length $\leq \ell$ in graph $G$, and it is computed using the recurrence $\mathbb{T}(v,\ell) = \underset{u\in V(G)\setminus \{v\}}{\min}\big(\mathbb{T}(u,\ell-1)+\operatorname{cost}(\{u,v\})\big)$. This algorithm can also be viewed as a circuit (over $\mathsf{R}^{+}$) of size $\mathcal{O}(n^3)$ and depth $\mathcal{O}(n)$ as follows: For each $\ell\in [n]$, the  circuit has a layer consisting of gates computing $\mathbb{T}(v,\ell)$'s for all $v\in V(G)$. The $v^{th}$ of these gates is a $\oplus$ (i.e., min) gate which is fed outputs of gates computing $\big(\mathbb{T}(u,\ell-1)+\operatorname{cost}(\{u,v\})\big)$'s for all $u\in V(G)\setminus \{v\}$. The $u^{th}$ of these gates is a $\otimes$ (i.e., $+$) gate which is fed $\mathbb{T}(u, \ell-1)$ (already computed by a lower layer) and the variable $x_{\{u,v\}}$. More generally, polynomial-time pure $(\min, +)$ DP algorithms (i.e., ones using only $\min$ and $+$ operations) can be viewed as polynomial-sized circuits over the corresponding min-plus semiring. So, lower bounds on size of $(\oplus, \otimes)$-circuits to compute a polynomial can be viewed as lower bounds on time needed by pure $(\min,+)$ DP algorithms to solve the corresponding problem.

In 2015, Jukna showed that the polynomials corresponding to Minimum Perfect Matching problem on bipartite graphs (i.e., Permanent polynomial) and Minimum Spanning Tree problem need $2^{\Omega(n)}$ sized circuits using only $\infty$ and $0$ as constants over $\mathsf{R}^{+}$ \cite{jukna2015lower}. So, although these problems can be solved in polynomial-time using Edmond's blossom algorithm (which works on non-bipartite graphs too) \cite{kolmogorov2009blossom} and Bor\d{u}vka's or Kruskal's or Prim's algorithm \cite{nevsetvril2001otakar, kruskal1956shortest, prim1957shortest} respectively, no pure $(\min, +)$ DP algorithm using only $\infty$ and $0$ as constants can solve them in polynomial (even sub-exponential) time over $\mathsf{R}^{+}$. In 1982, Jerrum and Snir showed that any circuit computing permanent needs $\geq n\cdot (2^{n-1}-1)$ many $\otimes$'s over $\mathsf{R}$ and $\mathsf{R}^{+}$ \cite{jerrum1982some}. So, even if pure $(\min, +)$ DP algorithms are allowed to use an unlimited number of $\min$ operations (i.e., contribution of $\min$'s to overall runtime is ignored), these algorithms remain less powerful than general algorithms.

It may be worth noting that the depth-reduction for circuits by Valiant, Skyum, Berkowitz and Rackoff \cite{valiant1983fast} has a semiring independent formulation (see Theorem 51 in \cite{jukna2015lower}) by Jukna.

\subsection{Incremental \texorpdfstring{$(\operatorname{min},+)$}{(min, +)} DP algorithms as ABPs}

The DP corresponding to Bellman-Ford algorithm is an incremental $(\min, +)$ algorithm, i.e., each $+$ operation has the input cost of an edge as one of its two arguments. That is, the corresponding $(\oplus, \otimes)$ circuit (described earlier) is skew, i.e., each $\otimes$ gate has an input variable as one of its two inputs. So, as skew circuits are equivalent to algebraic branching programs, this algorithm can be viewed as a width $\mathcal{O}(n)$ ABP (over $\mathsf{R}^{+}$) of size $\mathcal{O}(n^3)$ consisting of $\mathcal{O}(n)$ layers. Lower bounds on ABP size to compute a polynomial over min-plus semirings can be viewed as lower bounds on time needed by incremental $(\min, +)$ DPs to solve the corresponding problem.

\subsection{\texorpdfstring{$w$}{w}-Register Incremental \texorpdfstring{$(\operatorname{min},+)$}{(min, +)} DPs as Width-\texorpdfstring{$w$}{w} ABPs}
We can restrict every round/phase of the incremental DP algorithm to compute only a few (i.e., $w$) table entries using the $w$ table entries computed in the previous round. Such algorithms can be viewed as width-$w$ ABPs (i.e., ABPs wherein every layer has $\leq w$ nodes) over $\mathsf{R}$ or $\mathsf{R}^{+}$. Over fields, it is known that constant width ABPs can be efficiently simulated using formulas (see, for example, Proposition 7.1 in \cite{bringmann2018algebraic}); the same proof also works over $\mathsf{R}$ and $\mathsf{R}^{+}$. So, for any constant $w\geq 1$, $w$-register incremental $(\min, +)$ DP algorithms can be efficiently simulated using memoization-free $(\min, +)$ DP algorithms.

In 1988, Ben-Or \& Cleve also showed that width-$3$ ABPs can efficiently simulate (and so, are computationally equivalent to) formulas over fields \cite{cleve1988computing}. However, their construction is such that it works with additive inverses of polynomials; so, it does not work as is over semirings $\mathsf{R}$ and $\mathsf{R}^{+}$.

\subsection{Memoization-free Pure \texorpdfstring{$(\operatorname{min},+)$}{(min, +)} DP algorithms as Formulas}
A DP algorithm typically relies on the optimal substructure and overlapping subproblems properties of the problem to trade-off space for time. The first one says that the optimum solution of an instance can be found by combining optimum solutions for sub-instances. The second one says that the optimum solution of the same sub-instance may be needed to compute optimum solutions for multiple larger instances, and so memoization (i.e., storing its value so that it is readily available whenever needed, instead of recomputing it when needed) helps. Consider memoization-free pure $(\min, +)$ DP algorithms, i.e., pure $(\min, +)$ DPs wherein we restrict that the overlapping subproblems structure of the problem is not allowed to be exploited (i.e., any table entry computed in a round can be used to update only one other table entry in the next round). Such algorithms can be viewed as formulas (i.e., circuits wherein every gate has fan-out $\leq 1$) over $\mathsf{R}$ or $\mathsf{R}^{+}$. Valiant \cite{valiant1979completeness} showed that ABPs can efficiently simulate formulas over fields; this proof also works over $\mathsf{R}$ and $\mathsf{R}^{+}$. So, incremental $(\min, +)$ DPs can efficiently simulate memoization-free $(\min, +)$ DPs. 

In 1974, Brent showed that over rings, the depth of any formula can be reduced to logarithmic with only a polynomial size blow-up \cite{brent1974parallel}. In Appendix \ref{Brent's depth reduction N appendix}, we explain how his proof can be adapted to also work over $\mathsf{R}^{+}$.

\subsection{Related Work}

Jukna showed that there are polynomials that require exponential size for circuits \cite{jukna2015lower} and there are polynomials computable by poly-size circuits but require super-polynomial size ABPs \cite{DBLP:journals/orl/Jukna18}. Mahajan, Nimbhorkar, and Tawari \cite{mahajan2019shortest} showed lower bounds for constant-depth formulas with certain restrictions on fan-in computing the shortest path polynomial. Kluk and Nederlof \cite{kluk2025lower} showed that for any $k\geq 1 $, there is a graph on $k^{\mathcal{O}(1)}$ vertices such that any  $(\min,+)/(\max,+)$ circuit computing the corresponding Independent Set (resp. Travelling Salesman and Directed Spanning Tree) polynomial needs $\Omega(2^k)$ (resp. $2^{\Omega(k\log\log k)}$) gates. 

\subsection{Our Results}

Our main goal is to understand the computational power of all the four models discussed in the introduction and to establish a complexity-theoretic framework to study pure algorithms for optimization problems over $(\min, +)$ semirings.\\[-3mm]

\noindent{\bf Efficient Verifiability:} In Boolean and arithmetic complexity, the classes $\mathsf{NP}$ and $\mathsf{VNP}$ capture all efficiently verifiable problems. We define a complexity class in the $(\min, +)$ semiring setting that is a natural analogue of these classes. The straight-forward extension of the definition of $\mathsf{VNP}$ to semirings do not work as we can show that $\mathsf{VNP} = \mathsf{VP}$ with this definition. In Section \ref{VNP analogue section}, we define an analogue of $\textsf{VNP}$ over $\mathsf{R}$ and $\mathsf{R}^{+}$. Our definition allows summand to be a polynomial in complements of hypercube variables (apart from the original variables and hypercube variables); we discuss how, unlike fields, such a relaxation strengthens the obvious definition (i.e., without complements) over  $\mathsf{R}$ and $\mathsf{R}^{+}$. Intuitively, this happens because providing complemented bits compensates for the inability of a pure $(\min, +)$ DP/circuit over $\mathsf{R}$ and $\mathsf{R}^{+}$ to flip a bit (i.e., convert $0$ and $\infty$ to $\infty$ and $0$ respectively) on its own.

We have the following verifier-based interpretation of $\mathsf{VNP}$ over min-plus semirings: For any $f\in \mathsf{VNP}$, there exists $g\in \mathsf{VP}$ such that $f(X)$ is the minimum (i.e., $\oplus$) of $g(X, Y, \overline{Y})$ over all $\infty$-$0$ substitutions of $Y$ variables. Given any substitution $X^{*}$ of input variables $X$ and a constant $c$, suppose our goal is to check whether $f(X^{*}) \leq  c$. When $f(X^{*}) \leq c$, there is an $\infty$-$0$ substitution $Y^{*}$ of $Y$ variables that attained the minimum, i.e., for which $g(X^{*},Y^{*}, \overline{Y^{*}}) = f(X^{*}) \leq c$. View this substitution $Y^{*}$ (and its complement $\overline{Y^{*}}$) as a `certificate'. Also, if $f(X^{*}) > c$, then no matter what certificate (i.e., $\infty$-$0$ substitution $Y^{*}$ of $Y$ variables) is given, we always have $g(X^{*}, Y^{*}, \overline{Y^{*}}) \geq f(X^{*}) > c$. This way, we can view polynomial-sized circuit/polynomial-time pure $(\min, +)$ DP computing $g$ as a `verifier'. As such verifiers cannot flip bits, we restrict certificates to also contain flipped values of all bits.

We analyze the relationship between $\textsf{VP}_{\mathsf{R} (\text{or } \mathsf{R}^{+})}$ and $\textsf{VNP}_{\mathsf{R} (\text{or } \mathsf{R}^{+})}$ as the number of hypercube variables allowed to be complemented is varied in $\textsf{VNP}_{\mathsf{R} (\text{or } \mathsf{R}^{+})}$'s definition. We show the following dichotomy theorem (the super-script denotes the number of complemented hyper-cube variables):

\begin{restatable}[]{theorem}{VPVNPthm}
For semirings $S = \mathsf{R}$ or $S = \mathsf{R}^{+}$,
$\mathsf{VP}_{S} \neq \mathsf{VNP}_{S}^{[r(n)]}$ when $r(n)=\omega(\log n)$, and $\mathsf{VP}_{S} = \mathsf{VNP}_{S}^{[r(n)]}$ when $r(n)=\mathcal{O}(\log n)$, where $r(n)$ denotes the number of hypercube variables allowed to be complemented in the definition of $\mathsf{VNP}_S$.
\end{restatable}

Also, in Remark \ref{remark: ruling out exponential separation}, we discuss that the analysis involved in the proof of the above theorem also shows that $\mathsf{VP}_{\mathsf{R} (\text{or } \mathsf{R}^{+})}$ and $\mathsf{VNP}_{\mathsf{R} (\text{or } \mathsf{R}^{+})}^{[r(n)]}$ cannot be exponentially separated when $r(n) = o(n)$; in contrast, these are exponentially separated when $r(n) = \Omega(n)$. 

Similar to the case of fields, $\mathsf{VNP}_{\mathsf{R} (\text{or } \mathsf{R}^{+})}$ contains the permanent and Hamiltonian cycle families.  How\-ever, unlike fields (of characteristic $\neq 2$), a result by Grochow \cite{grochow2017monotone} implies that Hamiltonian cycle cannot be obtained as an efficient projection of permanent over $\mathsf{R}^{+}$; so, permanent is not $\mathsf{VNP}_{\mathsf{R}^{+}}$ complete. \\[-3mm]

\noindent{\bf ABP width as a resource:} 
A famous result in arithmetic circuit complexity by Ben-Or and Cleve shows that poly-size formulas and width-$3$ ABPs have the same computational power. The proof of this theorem has two main parts. First, Brent \cite{brent1974parallel} proved that any formula of size $s$ has an equivalent formula of depth $O(\log s)$. Second, Ben-Or and Cleve \cite{cleve1988computing} proved that a formula of depth-$d$ can be simulated by width-$3$ ABPs of size $O(\exp(d))$. 

One of the striking results in the arithmetic circuits over fields is that of Brent's depth reduction that shows that any size $s$ formula can be simulated by size $s^{O(1)}$ formula of depth $O(\log s)$. In Appendix \ref{Brent's depth reduction N appendix}, we adapt it to work over semiring $\mathsf{R}^{+}$. In the original proof, there is a critical a step that needs a polynomial to have additive inverse; so, this step does not work as is over min-plus semirings. Over $\mathsf{R}^{+}$, we circumvent this difficulty by making use of absorption in a certain way (see Appendix \ref{Brent's depth reduction N appendix} for details).
We leave the case of $\mathsf{R}$ as an open problem.
\begin{proposition}[Brent's Depth reduction for Formulas over $\mathsf{R}^{+}$]
\label{Adapting Brent's reduction over R+ proposition}
$~$\\
Any size $s$ formula over $\mathsf{R}^{+}$ can be simulated by a size $s^{\mathcal{O}(1)}$ formula of depth $\mathcal{O}(\log s)$.
\end{proposition}

It remains to investigate the power of constant width ABPs over semirings. First, we exhibit a simple polynomial that width-$2$ ABPs cannot compute.
\begin{restatable}[Non-universality of Width-2 ABPs]{theorem}{nonuniversalwidthtwosemiringN}
\label{non-universal width-2 semiring N}
$~$\\No width-$2$ ABP can compute $(x_1\otimes  y_1)\oplus (x_2\otimes y_2)\oplus (x_3\otimes y_3)$ over $\mathsf{R}$ and $\mathsf{R}^{+}$.
\end{restatable}

So, 2-register incremental $(\min, +)$ DPs cannot compute (even with unlimited runtime), for example, minimum weight 2-matching of $K_4$ (i.e., complete graph on four vertices), or minimum weight shortest path of $K_{2,3}$ (i.e., complete bipartite graph with parts of sizes two and three) from one vertex $s$ of the two-sized part to its other vertex $t$; see Figure \ref{K4 and K23 figure}.

\begin{figure}[ht!]
\centering
\includegraphics[scale=0.75]{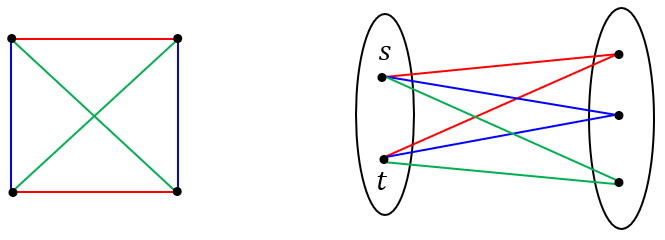}
\caption{: The left image shows $K_4$ and its three 2-matchings highlighted in red, blue and green. The right image shows $K_{2,3}$ and its three $s$-to-$t$ paths highlighted in red, blue and green.}
\label{K4 and K23 figure}
\end{figure}

This impossibility is proved using a cut-and-paste argument that is similar to some lower bounds for monotone arithmetic circuit models. We analyze how the paths producing the required monomials intersect with each other; at an intersection point, their sub-paths are patched together to identify a new path which produces a problematic monomial. For monotone arithmetic circuits, the problematic monomial can be any monomial that does not divide any monomial in the polynomial since the model does not allow cancellations. However, this is not sufficient for $(\min, +)$ semirings as monomials can be cancelled by absorption (e.g., $x \oplus (x\otimes y) = x$ over $\mathsf{R}^{+}$, $x^{\otimes 2} \oplus y^{\otimes 2} \oplus( x\otimes y) = x^{\otimes 2} \oplus y^{\otimes 2}$ over both $\mathsf{R}$ and $\mathsf{R}^{+}$). Over these semirings, to prove that a monomial is problematic, we have to show that the value of the monomial is less than the polynomial's value for some substitution of the variables by constants. Indeed, this is what we do in our impossibility proof. This polynomial has also been used to prove impossibility results for restricted width-$2$ ABPs over fields \cite{DBLP:conf/fsttcs/SahaSS09}.

 In contrast, we show that 3-register incremental $(\min, +)$ DPs can compute any polynomial. We prove the more general result:
 \begin{restatable}[Low Depth Formulas to low Width ABPs]{theorem}{lowdepthformulastolowwidthABPs}
\label{Low depth circuits to low width ABPs}
$~$\\For any $p\geq 1$, consider any depth $2p$ alternating $\bigoplus\bigotimes\ldots$ $\bigoplus\bigotimes$ formula $C$ of size $s$ over $\mathsf{R}$ or $\mathsf{R}^{+}$. Then, there is a width $(2p+1)$ ABP of size $\mathcal{O}(p\cdot s)$ that simulates the formula $C$.
\end{restatable}

Note that all polynomials are computable by formulas of depth $2$. By substituting $p = 1$ in the above theorem, we conclude that width-$3$ ABPs can compute any polynomial. In particular, minimum weight 2-matching of $K_n$ and minimum weight shortest $s$-to-$t$ path of $K_{2, n}$ have poly-size width-$3$ ABPs because the corresponding polynomials are sparse (i.e., have polynomial-sized $\bigoplus \bigotimes$ formulas). This illustrates a stark contrast between power of $2$-register and $3$-register incremental $(\min, +)$ DPs.

Interpreted for arbitrary $p$, the above results shows that any memoization-free alternating $(\min, +)$ DP algorithm that is highly parallelized (i.e., consists of only a few update rounds) can be efficiently simulated using an incremental $(\min, +)$ DP algorithm that updates only a few table entries per round.

Combined with a known DP algorithm for Shortest Path problem based on merging equal-length subpaths, Theorem \ref{Low depth circuits to low width ABPs} gives the following corollary (see details in Section \ref{low depth circuits to low width ABPs section}):
\begin{restatable}[Width-reduction for ABPs]{corollary}{ABPwidthreductioncorollary}
\label{ABPwidthreductioncorollary}
$~$\\ Any ABP of size $s$ over $\mathsf{R}$ or $\mathsf{R}^{+}$ can be converted into an equivalent width $(2p+1)$ ABP of size $s^{\mathcal{O}(p\cdot s^{\frac{1}{p}})}$ for any $p$. So, in particular, any ABP over $\mathsf{R}$ or $\mathsf{R}^{+}$ can be converted into an equivalent logarithmic width ABP with a quasi-polynomial blow-up in its size.
\end{restatable}

For the case of fields, we remark that the following width-reduction for ABPs can be inferred from known results: Consider an ABP of size $s=n^{\mathcal{O}(1)}$ computing a polynomial of degree $d=n^{\mathcal{O}(1)}$ over a field. Viewing this ABP as a skew-circuit and applying VSBR's circuit depth reduction \cite{valiant1983fast} gives a circuit of size $n^{\mathcal{O}(1)}$ and depth $\mathcal{O}((\log n)^2)$. Then, duplicating gates gives a formula of size $n^{\mathcal{O}((\log n)^2)}$. Next, applying Brent's formula depth reduction \cite{brent1974parallel} gives a formula of size $n^{\mathcal{O}((\log n)^2)}$ and depth $\mathcal{O}((\log n)^3)$. Then, applying Ben-Or and Cleve's result \cite{cleve1988computing} gives a width-3 ABP of size $2^{\mathcal{O} ((\log n)^3)}$. Our width-reduction over min-plus semirings (i.e., Corollary \ref{ABPwidthreductioncorollary}) is weaker because by bearing a quasi-polynomial size-blow up, it shrinks the width to $\mathcal{O}(\log n)$, not a constant. Finding/ruling-out an analogue of Ben-or and Cleve's result over min-plus semirings is our key open question; see Section \ref{sec: conclusion}.

In Section \ref{bivariate construction section}, we show that despite their non-universality, 2-register incremental $(\min, +)$ DP algorithms can at least compute all $\infty$-$0$ bivariate polynomials  (i.e., every monomial has coefficient $\infty$ or $0$) over $\mathsf{R}^{+}$ efficiently. The proof of this result crucially relies on absorption properties of semirings.

\begin{restatable}[Width-2 ABPs for $\infty$-$0$ Bivariate polynomials over $\mathsf{R}^{+}$]{theorem}{bivariateconstructiontheorem}
\label{uniform-coefficient bivariates over N}
$~$\\Any $\infty$-$0$ bivariate polynomial $f$ over $\mathsf{R}^{+}$ is computable by a size $\mathcal{O}(\operatorname{degree}(f))$ width-2 ABP.
\end{restatable}

Over fields, Bringmann, Ikenmeyer, and Zuiddam \cite{bringmann2018algebraic} showed that exponential hyper-cube sum over very weak models such as width-$2$ ABPs or products of linear forms are the same as $\mathsf{VNP}$. In sections \ref{VNP via hypercube sum over width-2 ABPs section} and \ref{hypercube sum over general width-1 ABPs section}, we show that these results hold over semirings as well.




\subsection{Proof Outlines}

In Section \ref{non-universality N section}, we show that width-2 ABPs cannot compute $(x_1\otimes  y_1)\oplus (x_2\otimes y_2)\oplus (x_3\otimes y_3)$ over $\mathsf{R}$ and $\mathsf{R}^{+}$. A high level description of this proof is as follows: First, we show that any ABP computing the above polynomial must have three paths, one for each monomial, such that each of these three paths contains both variables of the corresponding monomial exactly once each and no variable from the other two monomials. Then, we traverse the ABP from left-to-right and mark a layer based on the first variable occurrence amongst the edges of these three paths. Next, we traverse further to mark another layer based on the first variable occurrence amongst the edges of the remaining two paths. At this marked layer, we analyze the possible orientations of edges from the three paths entering into and exiting from it. For each possibility, we show a way to patch sub-paths from two of the three paths to form a new path; further, we analyze the possible weights of this new path and for each possibility, we show that the weight of this new path takes a value strictly lesser than the value taken by $(x_1\otimes y_1)\oplus (x_2\otimes y_2)\oplus (x_3\otimes y_3)$ for some choice of variable substitution, which is a contradiction. We make this argument precise in the proof of Theorem \ref{non-universal width-2 semiring N}.

In Section \ref{low depth circuits to low width ABPs section}, we show how to simulate any size $s$ depth $2p$ alternating $\bigoplus\bigotimes\ldots \bigoplus\bigotimes$ formula using a width $(2p+1)$ ABP of size $\mathcal{O}(p\cdot s)$ over $\mathsf{R}$ and $\mathsf{R}^{+}$. A high-level description of this proof is as follows: We associate a carefully chosen format matrix with each of the $2p$ levels of the formula. Moving bottom-to-top, we describe how to obtain width $(2p+1)$ ABPs that computes the outputs of nodes in any level $i$ (in format associated with level $i$) using the width $(2p+1)$ ABPs already built to compute outputs of nodes in level $i-1$ (in format associated with level $i-1$). We make this construction precise in the proof of Theorem \ref{Low depth circuits to low width ABPs}.

In Section \ref{bivariate construction section}, we show that despite their non-universality, width-2 ABPs can at least compute all $\infty$-$0$ (i.e., all monomials have coefficients $\infty$ or $0$) bivariate polynomials over $\mathsf{R}^{+}$. A high level description of this proof is as follows: If a monomial has both a higher $x$ power and a higher $y$ power than another monomial, then the former monomial always (i.e., for any possible variable subsitution) takes larger values than (and so, can be safely absorbed into) the latter monomial. After such absorptions, the monomials that remain are such that the higher their $x$ powers, the lower their $y$ powers. We construct a width-2 ABP that has all $x$'s along the top level and all $y$'s along the bottom level. The two levels are joined by suitably placed bridges. Any source-to-sink path first covers up some region of the bottom level, takes a bridge to jump to the top level, and then covers up some region of the top level. The earlier it takes the bridge jump, the lower its weight's $y$ power and higher its weight's $x$ power. Similarly, the later it takes the bridge jump, the higher its weight's $y$ power and lower its weight's $x$ power. This enables us to have a path for each monomial. We make this construction idea precise in the proof of Theorem \ref{uniform-coefficient bivariates over N}.

 
\section{Preliminaries}

Let $S$ be a set equipped with \emph{addition} (denoted as $+$) and \emph{multiplication} (denoted as $\cdot$) operations. Then, $(S,+, \cdot)$ is called a \emph{semiring} if i) $(S,+)$ is a commutative monoid, ii) $(S,\cdot)$ is a monoid, iii) multiplying (from either side) the additive identity with any element of $S$  gives back the additive identity, and iv) multiplication (from either side) distributes over addition. Based on the definition of \emph{monoid}, the first and second conditions together can be alternatively put as follows: i) $S$ is closed under $+$ and $\cdot$, ii) both $+$ and $\cdot$ are associative, iii) $+$ is commutative, and iv) there exist additive and multiplicative identities. Semirings generalize \emph{rings} (which, in turn, generalize \emph{fields} by allowing elements to not have multiplicative inverses) by allowing elements to not have additive inverses.\\[-2mm]

\noindent{\bf Min-plus semirings $\mathsf{R}$ and $\mathsf{R}^{+}$:}
We use $\mathsf{R}$ and $\mathsf{R}^{+}$ to denote the semirings $(\mathbb{R}\cup\{\infty\},\oplus,\otimes)$ and $(\mathbb{R}_{\geq 0}\cup\{\infty\},\oplus,\otimes)$ respectively, where $\mathbb{R}$ is the set of all real numbers, $\mathbb{R}_{\geq 0}$ is the set of all non-negative real numbers, $\oplus$ denotes the minimum operation (which serves as the semirings' addition operation), and $\otimes$ denotes usual addition (which serves as the semirings' multiplication operation). The multiplicative and additive identities of $\mathsf{R}$ and $\mathsf{R}^{+}$ are $0$ and $\infty$ respectively. Note that $\mathsf{R}$ and $\mathsf{R}^{+}$ are \emph{idempotent}, i.e., $a\oplus a = a$ for all $a\in \mathsf{R}$ (or $\mathsf{R}^{+}$). For any variable $x$ and any integer $d\geq 0$, we use $x^{\otimes d}$ to denote $\underbrace{x\otimes x\otimes \ldots \otimes x}_{d \mbox{ \small{times}}}$.\\[-2mm]

\noindent{\bf Polynomials over $\mathsf{R}$ and $\mathsf{R}^{+}$:} An $n$-variate polynomial $f(x_1, \ldots, x_n)$ over min-plus semiring $S$ ($=\mathsf{R}$ or $\mathsf{R}^{+}$) is of the form $\underset{a_1, \ldots, a_n\in \mathbb{N}}{\bigoplus}c_{a_1, \ldots, a_n}\otimes x_1^{\otimes a_1}\otimes x_2^{\otimes a_2}\otimes \ldots \otimes x_n^{\otimes a_n}$ (where the coefficients $c_{a_1, \ldots, a_n}$'s are elements of $S$, and only finitely many of them are $\neq \infty$), which naturally defines a function from $S^n$ to $S$ as follows: For every $(s_1, \ldots, s_n)\in S^{n}$, the corresponding function maps $(s_1, \ldots, s_n)$ to $f(x_1=s_1, \ldots, x_n=s_n)$, i.e., polynomial $f$'s value upon substituting variables $x_1, \ldots, x_n$ as $s_1, \ldots, s_n$ respectively. For any two polynomials $f$ and $g$ over $S$, if the value taken by $g$ is always (i.e., for each possible substitution of the variables from $S$) at least the value taken by $f$, then $f\oplus g$ defines the same function as $f$. So, $g$ can be \emph{absorbed} into $f$, i.e., we shall treat $f$ and $f\oplus g$ as same. For example, $x\oplus (x\otimes y) = x$ over $\mathsf{R}^{+}$ (but not $\mathsf{R}$), and $x^{\otimes 2} \oplus y^{\otimes 2} \oplus (x\otimes y) = x^{\otimes 2} \oplus y^{\otimes 2}$ over both $\mathsf{R}$ and $\mathsf{R}^{+}$. This is because $\min\{x, x+y\} = x$ for all $x, y\in \mathsf{R}^{+}$, and $\min\{2x, 2y, x+y\} = \min\{2x, 2y\}$ for all $x, y\in \mathsf{R}$ (or $\mathsf{R}^{+}$). Also, if two polynomials $f$ and $g$ over $S$ define the same function from $S^{n}$ to $S$, we shall treat them as equivalent to each other with respect to computation; that is, if a computational model (e.g., formula, circuit, ABP) computes $g$, then we shall also say that it computes $f$, and vice versa. 

It may be worth noting that the above examples can be generalized to get the following identities: i) For any polynomials $f$ and $g$ over $\mathsf{R}^{+}$, we have $f=f\oplus (f\otimes g)$, and ii) for any polynomials $f_1, \ldots, f_k$ over $\mathsf{R}$ (or $\mathsf{R}^{+}$) and any integer $\ell\geq 0$, we have $(f_1\oplus \ldots \oplus f_k)^{\otimes \ell} = f_1^{\otimes \ell}\oplus \ldots \oplus f_k^{\otimes \ell}$. In the second identity, any other term $f_1^{\otimes t_1}\otimes \ldots\otimes f_k^{\otimes t_k}$ (where $t_1+\ldots+t_k=\ell$) in multinomial expansion of LHS gets absorbed into $f_1^{\otimes \ell}\oplus \ldots \oplus f_k^{\otimes \ell}$. So, for instance, degree $\ell$ complete homogeneous symmetric polynomial, degree $\ell$ power sum symmetric polynomial and $\ell^{th}$ power of degree 1 elementary/complete homogeneous symmetric polynomial over $\mathsf{R}$ (or $\mathsf{R}^{+}$) define the same function.\\[-2mm]

\noindent{\bf $p$-families and $p$-projections:}
A \emph{$p$-family} is a sequence $(f_n)_{n\geq 1}$ of polynomials whose number of variables and degrees are \emph{polynomially-bounded}, i.e., there exist polynomial functions $a(n)$ and $ b(n)$ such that for each $n\geq 1$, $f_n$ has at most $a(n)$ variables and $\operatorname{degree}(f_n)\leq b(n)$. A $p$-family $(f_n)_{n\geq 1}$ is called a \emph{$p$-projection} of another $p$-family $(g_n)_{n\geq 1}$ if there exists a polynomial function $c(n)$ such that for each $n\geq 1$, $f_n$ can be obtained from $g_{c(n)}$ by substituting each of its variables by either a variable or a constant.\\[-2mm]

\noindent{\bf Algebraic Branching Programs - The weak, weakest and general models:}
An \emph{algebraic branching program} (ABP) is a layered directed acyclic graph with designated \emph{source} and \emph{sink} nodes. Every edge (directed from some layer to its next layer) is labelled with either a variable or a constant from the base semiring. The  ABP is said to compute the sum of weights of its source-to-sink paths, where \emph{weight} of a path is the product of its edges' labels. The maximum number of nodes per layer in an ABP is called its \emph{width}. A width-$k$ ABP can be viewed as a sequence of $k\times k$ matrices, and vice versa. To do so, set the entry at $i^{th}$ row and $j^{th}$ column of $\ell^{th}$ matrix as the label of the edge directed from $i^{th}$ vertex of layer $\ell$ to $j^{th}$ vertex of layer $\ell+1$. Then, $(1,1)$-entry of the product of these $k\times k$ matrices is same as the polynomial computed by the width-$k$ ABP.\\[-3mm]

The ABPs, as defined above, are often called \emph{weakest} ABPs. Two stronger definitions, called \emph{weak} ABPs and \emph{general} ABPs, allow for edge labels to be linear forms in one variable and linear forms respectively. It may be worth noting that as described in \cite{univariatesemiring}, any univariate polynomial over $\mathsf{R}$ can be written as a product of linear factors based on its \emph{tropical roots/break points}, i.e., points at which slope of the univariate polynomial (viewed as a function) changes; equivalently, these are the points where the minimum is attained by at least two monomials. See Figure \ref{tropical roots figure} for an example. So, any univariate polynomial over $\mathsf{R}$ can be computed using a weak/general width-$1$ ABP. Also, any univariate polynomial can be computed by a weakest width-2 ABP (see Figure \ref{width 2 univariate figure}). In this paper, we mostly work with weakest ABPs (except in Section \ref{hypercube sum over general width-1 ABPs section}, which considers general ABPs).

\begin{appsection}{Example illustrating Factoring of Univariate Tropical polynomials}{appsec:example-tropical-factoring}
\begin{figure}[ht!]
\centering
\includegraphics[scale=0.7]{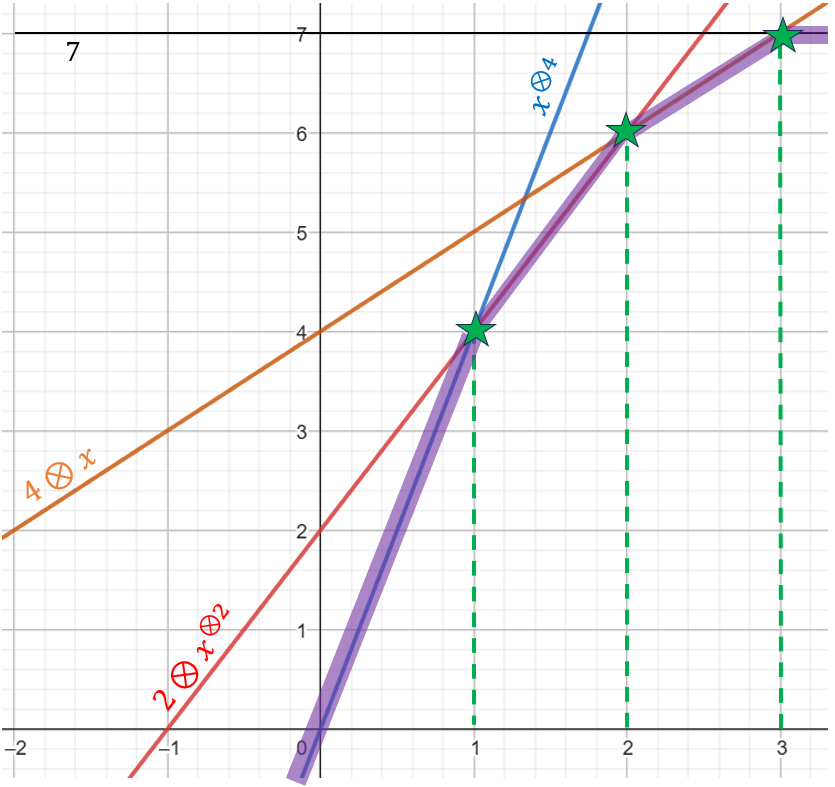}
\caption{: Consider the univariate polynomial $p(x):=x^{\otimes 4}\oplus (2\otimes x^{\otimes 2})\oplus (4\otimes x)\oplus 7$ over $\mathsf{R}$. The plots of the monomials $x^{\otimes 4}$, $2\otimes x^{\otimes 2}$, $4\otimes x$ and $7$ are in blue, red, orange and black respectively. The plot of $p(x)$ is their minimum (in purple). The tropical roots/break points are at $x=1$, $2$ and $3$, as marked by green stars. Observe that  $p(x)$ can be factored as $(x\oplus 1)^{\otimes 2}\otimes (x\oplus 2)\otimes (x\oplus 3)$.}
\label{tropical roots figure}
\end{figure} 
\end{appsection}

\begin{appsection}
{Weakest width-2 ABPs computing Univariate polynomials}{appsec:width2-ABP-univariate-poly}
\begin{figure}[ht!]
\centering
\includegraphics[scale=0.75]{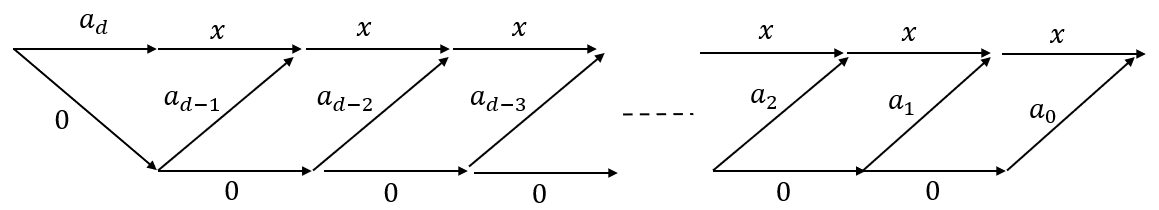}
\caption{: This figures shows a weakest width-2 ABP of size $\mathcal{O}(d)$ that computes a degree $d$ univariate polynomial $p(x) = (a_d\otimes x^{\otimes d})\oplus (a_{d-1}\otimes x^{\otimes (d-1)})\oplus \ldots \oplus (a_2\otimes x^{\otimes 2}) \oplus (a_1 \otimes x)\oplus a_0$.}
\label{width 2 univariate figure}
\end{figure}
\end{appsection}

\noindent{\bf Permanent and Hamiltonian cycle families over $\mathsf{R}$ and $\mathsf{R}^{+}$:} The \emph{permanent family} is $(perm_n)_{n\geq 1}$, where $perm_n$ denotes the permanent of an $n\times n$ matrix whose entries are $n^2$ variables $x_{i,j}\vert_{1\leq i, j\leq n}$'s. That is, 
$perm_n := \underset{\sigma \in S_n}{\bigoplus} x_{1, \sigma(1)}\otimes x_{2, \sigma(2)}\otimes \ldots \otimes x_{n, \sigma(n)}$, 
where $S_n$ denotes the set of all permutations of $[n]$. It corresponds to minimum weight perfect matching of complete bipartite graph $K_{n, n}$ and minimum weight cycle cover of $n$-vertex complete directed graph with self-loops. The \emph{Hamiltonian cycle family} is $(HC)_{n\geq 1}$, where $HC_n:= \underset{\sigma\in C_n}{\bigoplus}x_{1, \sigma_1}\otimes x_{2, \sigma(2)}\otimes \ldots \otimes x_{n, \sigma(n)}$ and $C_n\subseteq S_n$ denotes the set of all cyclic permutations of $[n]$. It corresponds to minimum weight Hamiltonian cycle of an $n$-vertex complete directed graph.
\section{$\mathsf{VP}\stackrel{?}{=}\mathsf{VNP}$ and its Analogue over $\mathsf{R}$ and $\mathsf{R}^{+}$}
\label{VNP analogue section}

Recall that over any field $\mathbb{F}$, the class $\mathsf{VNP}_{\mathbb{F}}$ consists of all $p$-families expressible as a hypercube sum over a family in $\mathsf{VP}_{\mathbb{F}}$ using polynomially-many hypercube variables which take Boolean values $0$ and $1$. Note that the additive and multiplicative identities of $\mathsf{R}$ and $\mathsf{R}^{+}$ are $\infty$ and $0$ respectively. So, an obvious attempt to define $\mathsf{VNP}$ over $\mathsf{R}$ and $\mathsf{R}^{+}$ would be to keep the same definition except that the hypercube variables take Boolean values $\infty$ and $0$ (instead of $0$ and $1$). That is, $\mathsf{VNP}_{\mathsf{R} (\text{or } \mathsf{R}^{+})}$ would consist of all $p$-families $(f_n)_{n\geq 1}$ for which there exists a $p$-family $(g_n)_{n\geq 1}\in \mathsf{VP}_{\mathsf{R} (\text{or } \mathsf{R}^{+})}$ and polynomials $p(n)$
\& $q(n)$ such that $f_n(X) = \underset{y_1, \ldots, y_{q(n)}\in \{\infty, 0\}}{\bigoplus}g_{p(n)}(X, y_1, \ldots, y_{q(n)})$, where $y_1, \ldots, y_{q(n)}$ are hypercube variables. Now, view $g_{p(n)}$ as a polynomial in $y_1, \ldots, y_{q(n)}$ whose coefficients are polynomials in $X$. That is, $g_{p(n)}$ is of the form $\underset{a_1, \ldots, a_{q(n)}\in \mathbb{N}}{\bigoplus}h_{a_1, \ldots, a_{q(n)}}(X)\otimes y_1^{\otimes a_1}\otimes \ldots \otimes y_{q(n)}^{\otimes a_{q(n)}}$ for some polynomials $h_{a_1, \ldots, a_{q(n)}}$'s in $X$. Plugging this $g_{p(n)}$'s expression in $f_n$'s expression,  
\begin{equation*}
f_n(X) = \underset{y_1, \ldots, y_{q(n)}\in \{\infty, 0\}}{\bigoplus}\Big(\underset{a_1, \ldots, a_{q(n)}\in \mathbb{N}}{\bigoplus}h_{a_1, \ldots, a_{q(n)}}(X)\otimes y_1^{\otimes a_1}\otimes \ldots \otimes y_{q(n)}^{\otimes a_{q(n)}}\Big).
\end{equation*}
Changing the order of the two summations (i.e., actually, two minimums), we get  
\begin{equation*}
\begin{split}
f_n(X) & = \underset{a_1, \ldots, a_{q(n)}\in \mathbb{N}}{\bigoplus}h_{a_1, \ldots, a_{q(n)}}(X)\otimes \underbrace{\Big(\underset{y_1, \ldots, y_{q(n)}\in \{\infty, 0\}}{\bigoplus} y_1^{\otimes a_1}\otimes \ldots \otimes y_{q(n)}^{\otimes a_{q(n)}}\Big)}_{=0}\\
& = \underset{a_1, \ldots, a_{q(n)}\in \mathbb{N}}{\bigoplus}h_{a_1, \ldots, a_{q(n)}}(X) ~~= g_{p(n)}(X, 0, \ldots, 0),
\end{split}
\end{equation*}
which has polynomial-sized circuit; so, this definition of $\mathsf{VNP}_{\mathsf{R} (\text{or } \mathsf{R}^{+})}$ makes it coincide with $\mathsf{VP}_{\mathsf{R} (\text{or } \mathsf{R}^{+})}$.

\subsection{Allowing Complements of Hypercube variables}
Over any field $\mathbb{F}$, suppose that we modify the definition of $\mathsf{VNP}_{\mathbb{F}}$ to say that the hypercube sum's summand is allowed to be a polynomial in the complements of hypercube variables (apart from the original variables and hypercube variables), where the complement $\overline{y}$ of any hypercube variable $y$ takes values $0$ and $1$ when $y$ takes values $1$ and $0$ respectively. However, since $\overline{y}=1-y$ for any hypercube variable $y$, this modified definition of $\mathsf{VNP}_{\mathbb{F}}$ is same as its original definition. Now, let us mimic this modification over $\mathsf{R}$ and $\mathsf{R}^{+}$ as follows: Define $\mathsf{VNP}_{\mathsf{R} (\text{or } \mathsf{R}^{+})}$ to consist of those $p$-families $(f_n)_{n\geq 1}$ for which there exists a $p$-family $(g_n)_{n\geq 1}\in \mathsf{VP}_{\mathsf{R} (\text{or } \mathsf{R}^{+})}$ and polynomials $p(n)$ \& $q(n)$ such that $f_n(X) = \underset{y_1, \ldots, y_{q(n)}\in \{\infty, 0\}}{\bigoplus}g_{p(n)}(X, y_1, \ldots, y_{q(n)}, \overline{y_1}, \ldots, \overline{y_{q(n)}})$, where the comple\-ment $\overline{y}$ of any hypercube variable $y$ takes values $\infty$ and $0$ when $y$ takes values $0$ and $\infty$ respectively. Unlike fields (where $\overline{y}=1-y$), $\overline{y}$ cannot be realized as a polynomial expression in $y$ over $\mathsf{R}$ and $\mathsf{R}^{+}$. So, it is conceivable that this modified definition may strictly strengthen the obvious definition of $\mathsf{VNP}_{\mathsf{R} (\text{or } \mathsf{R}^{+})}$ attempted above (i.e., separate it from $\mathsf{VP}_{\mathsf{R} (\text{or } \mathsf{R}^{+})}$). To show that this is indeed true, we prove that the permanent family (which needs $2^{\Omega(n\log n)}$ sized circuits over $\mathsf{R}$ and $\mathsf{R}^{+}$ \cite{jerrum1982some} and so, $\not\in \mathsf{VP}_{\mathsf{R} (\text{or } \mathsf{R}^{+})}$) satisfies this definition as follows: Recall that $perm_{n} := \underset{\sigma\in S_{n}}{\bigoplus} \underset{1\leq i\leq n}{\bigotimes}x_{i, \sigma(i)}$. Note that every map $\sigma$ from $[n]$ to $[n]$ can be specified by a $\infty$-$0$  matrix of size $n\times \log n$, where for every $1\leq i\leq n$, its $i^{th}$ row indicates binary encoding of $i$'s image under $\sigma$. Also, all permutations (i.e., bijective maps) correspond to those such matrices wherein all rows are distinct. So, we write
\begin{equation*}
\begin{split}
perm_{n} & = \underset{Y\in \{\infty, 0\}^{n\times \log n}}{\bigoplus}\big[\substack{\text{All rows of $Y$}\\ \text{are distinct}}\big]\otimes \bigotimes_{i=1}^{n}\Bigg(\bigoplus_{t=1}^{n}x_{i,t}\otimes \Big[\substack{\big(Y_{i,1}, \ldots, Y_{i, \log n}\big) \text{ is} \\ \text{the binary encoding of } t}\Big]\Bigg)\\
& = \underset{Y\in \{\infty, 0\}^{n\times \log n}}{\bigoplus}\Bigg(\underset{1\leq u<v\leq n}{\bigotimes}\bigoplus_{w=1}^{\log n}[Y_{u,w}\neq Y_{v,w}]\Bigg)\otimes \bigotimes_{i=1}^{n}\Bigg(\bigoplus_{t=1}^{n}x_{i,t}\otimes \bigotimes_{w=1}^{\log n}[Y_{i,w}=t_w]\Bigg),
\end{split}
\end{equation*}
where for every $1\leq t\leq n$, $(t_1, \ldots, t_{\log n})$ denotes the binary encoding of $t$, and $[\ldots]$ denotes the indicator function (i.e., it takes values $0$ and $\infty$ when the statement enclosed within $[\ldots]$ is true and false respectively). Thus, we get the following expression for $perm_{n}$:
\begin{equation*}
\hspace{-1.1 cm}
\underset{Y\in \{\infty, 0\}^{n\times \log n}}{\bigoplus}\underbrace{\Bigg(\underset{1\leq u<v\leq n}{\bigotimes}\bigoplus_{w=1}^{\log n}\big(\overline{Y_{u,w}}\otimes  Y_{v,w}\oplus Y_{u,w} \otimes  \overline{Y_{v,w}}\big)\Bigg)\otimes  \bigotimes_{i=1}^{n}\Bigg(\bigoplus_{t=1}^{n}x_{i,t}\otimes \bigotimes_{w=1}^{\log n}\big(t_w\otimes  Y_{i,w}\oplus \overline{t_w}\otimes  \overline{Y_{i,w}}\big)\Bigg)}_{\text{This has $n^{\mathcal{O}(1)}$ sized formula, as desired.}}
\end{equation*}

\subsection{Spectrum of number of Complementable Hypercube variables - A Dichotomy Theorem}
While our first attempt to define $\mathsf{VNP}_{\mathsf{R} (\text{or } \mathsf{R}^{+})}$ allowed no hypercube variable to be complemented, the modified definition above allowed all the polynomially-many hypercube variables to be compl\-emented. Between these two extremes, lies a spectrum of definitions based on how many hypercube variables are allowed to be complemented. For any function $r(n)$, let $\mathsf{VNP}_{\mathsf{R} (\text{or } \mathsf{R}^{+})}^{[r(n)]}$ denote the class corresponding to the definition that allows $r(n)$ many hypercube variables to be complemented. We analyze how the relationship between $\mathsf{VP}_{\mathsf{R} (\text{or } \mathsf{R}^{+})}$ and $\mathsf{VNP}_{\mathsf{R} (\text{or } \mathsf{R}^{+})}^{[r(n)]}$ changes as $r(n)$ is varied. In particular, we show the following dichotomy theorem.

\VPVNPthm*


First, we show that $\mathsf{VP}_{\mathsf{R} (\text{or } \mathsf{R}^{+})} \neq \mathsf{VNP}_{\mathsf{R} (\text{or } \mathsf{R}^{+})}^{[r(n)]}$ when $r(n)=\omega(\log n)$ as follows: Consider the family $(f_n)_{n\geq 1}$ defined as $f_n:= perm_{\beta(n)}$, where $\beta(n)$ denotes the function for which $\beta(n) \log(\beta(n)) = r(n)$. As described earlier, $perm_{\beta(n)}$ can be written as a hypercube sum using $\beta(n) \log(\beta(n))$ hypercube variables and their complements. Also, $perm_{\beta(n)}$ needs $2^{\Omega\big(\beta(n)\log(\beta(n))\big)} = 2^{\omega(\log n)} = n^{\omega(1)}$ sized circuits over $\mathsf{R}$ and $\mathsf{R}^{+}$ \cite{jerrum1982some}. Therefore, $(f_n)_{n\geq 1}\in \mathsf{VNP}_{\mathsf{R} (\text{or } \mathsf{R}^{+})}^{[r(n)]} \setminus \mathsf{VP}_{\mathsf{R} (\text{or } \mathsf{R}^{+})}$, as desired. 

Next, we show that $\mathsf{VP}_{\mathsf{R} (\text{or } \mathsf{R}^{+})} = \mathsf{VNP}_{\mathsf{R} (\text{or } \mathsf{R}^{+})}^{[r(n)]}$ when $r(n)=\mathcal{O}(\log n)$ as follows: Consider any family $(f_n)_{n\geq 1}\in \mathsf{VNP}_{\mathsf{R} (\text{or } \mathsf{R}^{+})}^{[r(n)]}$. Then, there exists a family $(g_n)_{n\geq 1}\in \mathsf{VP}_{\mathsf{R} (\text{or } \mathsf{R}^{+})}$ and polynomials $p(n)$ \& $q(n)$ such that $f_n(X) = \underset{y_1, \ldots, y_{q(n)}\in \{\infty, 0\}}{\bigoplus}g_{p(n)}(X, y_1, \ldots, y_{q(n)}, \overline{y_1}, \ldots, \overline{y_{c \log (n)}})$. Now, view $g_{p(n)}(X, y_1, \ldots, y_{q(n)},$$ \overline{y_1}, \ldots, \overline{y_{c\log(n)}})$ as a polynomial in $y_1, \ldots, y_{q(n)},$ $ \overline{y_1}, \ldots,$ $ \overline{y_{c\log(n)}}$ whose coefficie\-nts are polynomials in $X$. That is, $g_{p(n)}$ is of the following form:
$$\underset{\substack{\underline{a}=(a_1, \ldots, a_{q(n)})\in \mathbb{N}^{q(n)}\\\underline{b}=(b_1, \ldots, b_{c\log(n)})\in \mathbb{N}^{c\log(n)}}}{\bigoplus}h_{\underline{a},\underline{b}}(X)~\otimes \bigotimes_{i=1}^{q(n)}y_i^{\otimes a_i} \otimes \bigotimes_{j=1}^{c\log(n)}(\overline{y_j})^{\otimes b_j},$$
for some polynomials $h_{a_1, \ldots, a_{q(n)}, b_1, \ldots, b_{c\log(n)}}$'s in $X$. Also, any summand wherein some hypercube variable and its complement both have a positive power vanishes. That is, if both $a_i$ and $b_{i}$ are $>0$ for some $1\leq i\leq c \log(n)$, then the summand evaluates to $\infty$. This is because for such a summand, $y_i^{a_i}=\infty$ when $y_i=\infty$, and $(\overline{y_i})^{b_i}=\infty$ when $y_i=0$. So, all such summands can be safely dropped. So, we are left with only those summands wherein for each $1\leq i\leq c\log(n)$, only one of $y_i$ and $\overline{y_i}$ appears. Thus, we can express $g_{p(n)}(X, y_1, \ldots, y_{q(n)},\overline{y_1}, \ldots, \overline{y_{c \log(n)}})$ in the following form:
\begin{equation*}
\hspace{-0.6 cm}\underset{\substack{\underline{\sigma}=(\sigma_1,\ldots, \sigma_{c\log(n)}) \in \{\text{comp},  \text{ not-comp}\}^{c\log(n)}\\\underline{\lambda}=(\lambda_1, \ldots, \lambda_{q(n)})\in \mathbb{N}^{q(n)}}}{\bigoplus}  \tilde{h}_{\substack{\underline{\sigma}, \underline{\lambda}}}(X)~\otimes\bigotimes_{i=1}^{c\log(n)}(\sigma_i(y_i))^{\otimes \lambda_i}\otimes \bigotimes_{j=c\log(n)+1}^{q(n)}y_j^{\otimes \lambda_j},
\end{equation*}
for some polynomials $\tilde{h}$'s in $X$, and where for each $1\leq i \leq c\log(n)$, $\sigma_i(y_i) := y_i$ when $\sigma_i=\text{not-comp}$, and $\sigma_i(y_i) :=\overline{y_i}$ when $\sigma_i=\text{comp}$. Next, plugging the above expression of 
$g_{p(n)}$ in the expression of $f_n$, and then changing the order of the two summations (i.e., actually, two minimums), we get the following expression for $f_n(X)$:
\begin{equation*}
\begin{split}
& \underset{\underline{\sigma}, \underline{\lambda}}{\bigoplus} ~\tilde{h}_{\underline{\sigma}, \underline{\lambda}}(X)\otimes \underbrace{\underset{y_1, \ldots, y_{q(n)}\in \{\infty,0\}}{\bigoplus}~\bigotimes_{i=1}^{c\log(n)}\big(\sigma_i(y_i)\big)^{\otimes {\lambda_i}}\otimes \bigotimes_{j=c\log(n)+1}^{q(n)} y_j^{\otimes \lambda_j}}_{A_{\underline{\sigma}, \underline{\lambda}}.}
\end{split}
\end{equation*}
Note that $A_{\underline{\sigma}, \underline{\lambda}} = 0$ because the summand corresponding to  $y_i=\sigma_i(0)$ for all $1\leq i\leq c\log(n)$ and $y_{j}=0$ for all $c\log(n)+1\leq j\leq q(n)$ is $0$; here, for each $1\leq i\leq c\log(n)$, $\sigma_{i}(0)$ is defined as 0 when $\sigma_i=$ not-comp, and $\infty$ when  $\sigma_i =$ comp. So, we get
\begin{equation}
f_n(X) = \underset{\underline{\sigma}}{\bigoplus}\Bigg(\underset{\underline{\lambda}}{\bigoplus}~\tilde{h}_{\underline{\sigma}, \underline{\lambda}}(X)\Bigg).
\end{equation}
For every $\underline{\sigma}\in \{\text{comp, not-comp}\}^{c\log(n)}$, let $g_{p(n)}\vert_{\Phi_{\underline{\sigma}}}$ denote the polynomial obtained from $g_{p(n)}$ by substituting $y_i$ as $\sigma_i(0)$ for all $1\leq i\leq c\log(n)$ and $y_j$ as $0$ for all $c\log(n)< j\leq q(n)$.
Then, note that
\begin{equation}
\begin{split}
&\underset{\underline{\sigma}}{\bigoplus}~ g_{p(n)}\Bigg\vert_{\Phi_{\underline{\sigma}}} = ~\underset{\underline{\sigma}}{\bigoplus}~\Bigg(\underset{\underline{\sigma'}, \underline{\lambda}}{\bigoplus} ~\tilde{h}_{\underline{\sigma'}, \underline{\lambda}}(X)~\otimes \bigotimes_{i=1}^{c\log(n)} ~\big(\sigma'_i(\sigma_i(0))\big)^{{\otimes \lambda_i}}\Bigg)\\
& = \underbrace{\Bigg(\underset{\underline{\sigma}}{\bigoplus}\Bigg(\underset{\underline{\lambda}}{\bigoplus}~\tilde{h}_{\underline{\sigma}, \underline{\lambda}}(X)\Bigg)\Bigg)}_{\text{\normalsize corresponding to } \underline{\sigma}'=\underline{\sigma}} \oplus \Bigg(\substack{\text{\normalsize Some summands corresponding to } \underline{\sigma}'\neq \underline{\sigma}\\ \text{ \normalsize may survive. However, any such summand}\\ \text{\normalsize must have already appeared in the previous}\\ \text{\normalsize term and so, it can be safely ignored.}}\Bigg)\\
& = \underset{\underline{\sigma}}{\bigoplus}\Bigg(\underset{\underline{\lambda}}{\bigoplus}~\tilde{h}_{\underline{\sigma}, \underline{\lambda}}(X) \Bigg).
\end{split}
\end{equation}
Therefore, using (1) and (2), we get $f_n(X)=\underset{\underline{\sigma}}{\bigoplus} ~g_{p(n)}\big\vert_{\Phi_{\underline{\sigma}}}$. So, as $g_{p(n)}\vert_{\Phi_{\underline{\sigma}}}$ has an $n^{\mathcal{O}(1)}$ sized circuit for each $\underline{\sigma}\in \{\text{comp, not-comp}\}^{c\log(n)}$, it follows that $f_n(X)$ can be computed by a circuit of size $2^{c\log(n)}\cdot n^{\mathcal{O}(1)} = n^{\mathcal{O}(1)}$. Thus, we have $(f_n)_{n\geq 1}\in \mathsf{VP}_{\mathsf{R} (\text{or } \mathsf{R}^{+})}$, as desired.

\begin{remark}
\label{remark: ruling out exponential separation}
We showed above that when the number of complementable hypercube variables $r(n)=\mathcal{O}(\log n)$, families in $\VNP_{\mathsf{R} (\text{or } \mathsf{R}^{+})}^{[r(n)]}$ admit circuits of size $2^{\mathcal{O}(\log n)} \cdot n^{\mathcal{O}(1)}$. More generally, for any $r(n)$, the same argument gives circuits of size $2^{r(n)} \cdot n^{\mathcal{O}(1)}$. So, when $r(n) = o(n)$,  \VP$_{\mathsf{R} (\text{or } \mathsf{R}^{+})}$ and \VNP$_{\mathsf{R} (\text{or } \mathsf{R}^{+})}^{[r(n)]}$ cannot be exponentially separated. In contrast, when $r(n)=\Omega(n)$, \VP$_{\mathsf{R} (\text{or } \mathsf{R}^{+})}$ and \VNP$_{\mathsf{R} (\text{or } \mathsf{R}^{+})}^{[r(n)]}$ are exponentially separated as $(perm_{\beta(n)})_{n\geq 1}$ (where $\beta(n)$ is defined by $\beta(n) \log(\beta(n))=r(n)$) is in \VNP$_{\mathsf{R} (\text{or } \mathsf{R}^{+})}^{[r(n)]}$ but needs $2^{\Omega\big(\beta(n)\log(\beta(n))\big)} = 2^{\Omega(n)}$ sized circuits.
\end{remark}


\subsection{$\mathsf{VNP}_{\mathsf{R} \text{(or }\mathsf{R}^{+})}$ via Hypercube sum over Width-2 ABPs}
\label{VNP via hypercube sum over width-2 ABPs section}

The proof of $\mathsf{VP} \subseteq \mathsf{VNF}$ (and so, $\mathsf{VNP} = \mathsf{VNF}$) over fields in \cite{malod2008characterizing} works almost as is over $\mathsf{R}$ and $\mathsf{R}^{+}$ too (explained in Appendix \ref{VNP=VNF adapted proof appendix}). For any $p$-family $(f_n)_{n\geq 1} \in \mathsf{VP}_{\mathsf{R} (\text{or } \mathsf{R}^{+})}$, this  proof converts circuit computing $f_n$ into an equivalent polynomial-sized multiplicatively disjoint circuit $C_n$, expresses $f_n$ as the sum of values of parse trees of $C_n$, and then re-expresses this sum in terms of indicator variables (which serve as hypercube variables) to go over all subgraphs of $C_n$, using certain indicators to ensure that only the summands corresponding to parse trees survive. This gives $$f_n(X) = \underset{\substack{\underline{p}\in \{\infty, 0\}^{|V(C_n)|}\\ \underline{a}\in \{\infty, 0\}^{|E(C_n)|}}}{\bigoplus}g_n(X, \underline{p}, \underline{a}, \overline{\underline{p}}, \overline{\underline{a}})$$ where $p_v\mid_{v\in V(C_n)}$, $a_{(u,v)}\mid_{(u,v)\in E(C_n)}$ are hypercube variables, 
\begin{equation*}
g_n:= \Big(\underset{\substack{(u,v) \\\in E(C_n)}}{\bigotimes}B_{(u,v)}\Big)\otimes  p_{\operatorname{root}}\otimes  \Big(\underset{\substack{u \text{ is a } \\ \otimes \text{ gate}}}{\bigotimes}C_{u}\Big) \otimes  \Big(\underset{\substack{u \text{ is a } \\\oplus \text{ gate}}}{\bigotimes}D_{u}\Big) \otimes  \Big(\underset{\substack{u \ne \operatorname{root}}}{\bigotimes   }E_u\Big) \otimes  \Big(\underset{u\in \text{leaves}(C_n)}{\bigotimes}A_{u}\Big),
\end{equation*}

Let $\ell(u)$ (respectively $r(u)$) be the left child (resp. right child) of $u$ in the circuit $C_n$.
$$B_{(u,v)}:= \overline{a_{(u,v)}}\oplus a_{(u,v)}\otimes p_u \otimes p_v,~~~~ C_u := \overline{p_u} \oplus p_u \otimes a_{(\ell(u), u)}  \otimes a_{(r(u), u)}$$ 
$$D_u:= \overline{p_u}\oplus~ p_u\otimes (a_{(\ell(u), u)}\otimes \overline{a_{(r(u),u)}} \oplus a_{(r(u), u)}\otimes \overline{a_{(\ell(u), u)}}),$$ 
$$E_u:= \overline{p_u}\oplus    p_u\otimes  \underset{v~:~(u,v) \in E(C_n)}{\bigoplus   }a_{(u,v)}~\text{and}~ A_{u} := \overline{p_{u}} \oplus \text{label}(u)\otimes p_{u}.$$

As $B_{(u,v)}$'s, $C_u$'s, $D_u$'s, $E_u$'s and $A_{u}$'s have polynomial-sized formulas, so does $g_n$. Let us show that each of these also has a polynomial-sized width-2 ABP. Then, concatenating (to get the product of) all these polynomially-many width-2 ABPs would give a polynomial-sized width-2 ABP that computes $g_n$, thereby strengthening $\mathsf{VNP}_{\mathsf{R} (\text{or }\mathsf{R}^{+})} = \mathsf{VNF}_{\mathsf{R} (\text{or }\mathsf{R}^{+})}$ to $\mathsf{VNP}_{\mathsf{R} (\text{or }\mathsf{R}^{+})} = \mathsf{VNBP_2}_{~\mathsf{R} (\text{or }\mathsf{R}^{+})}$, where $\mathsf{VNBP_2}_{~\mathsf{R} (\text{or }\mathsf{R}^{+})}$ consists of families expressible as hypercube sum over a family with polynomial-sized width-2 ABPs over $\mathsf{R}$ (or $\mathsf{R}^{+}$). Note that $B_{(u,v)}$'s, $C_u$'s, $A_{u}$'s can be computed using width-2 ABPs shown in Figure \ref{label-Malod1}. Also, $D_u$'s can be computed using width-2 ABP shown in Figure \ref{label-Malod2}. The ABP in Figure \ref{label-Malod2} actually computes
\begin{figure}[h!]
\centering
\includegraphics[scale=0.6]{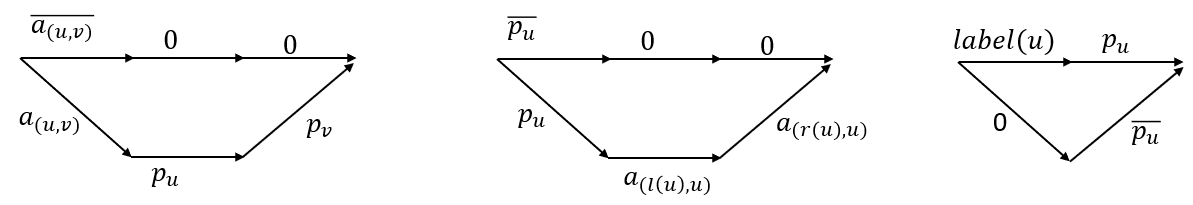}
\caption{Width-2 ABPs computing $B_{(u,v)}$'s, $C_u$'s, $A_{u}$'s.}
\label{label-Malod1}
\end{figure}
\begin{figure}[h!]
\centering
\includegraphics[scale=0.65]{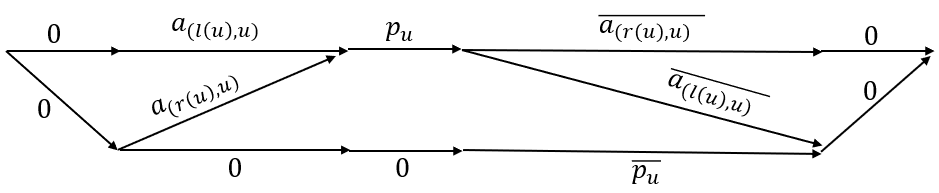}
\caption{Width-2 ABP computing $D_u$.}
\label{label-Malod2}
\end{figure} 
$$D_u \oplus p_u\otimes \big((a_{(\ell(u),u)}\otimes \overline{a_{(\ell(u),u)}}) \oplus (a_{(r(u),u)}\otimes \overline{a_{(r(u),u)}})\big),$$ but the extra terms $a_{(\ell(u),u)}\otimes \overline{a_{(\ell(u),u)}}$ and $a_{(r(u),u)}\otimes \overline{a_{(r(u),u)}}$ are not problematic as their value is $\infty$ for all possible $\infty$-$0$ substitutions of $a_{(\ell(u),u)}$ and $a_{(r(u),u)}$. Next, we show that each of the $E_u$'s can be computed by a polynomial-sized width-2 ABP. We have
\begin{equation*}
\begin{split}
E_u & := \overline{p_u}\oplus    p_u\otimes  \underset{v:~(u,v)\in E(C_n)}{\bigoplus   }a_{(u,v)}
\\
& = \underset{\substack{
v: ~(u,v)\in E(C_n)}}{\bigoplus   }\big(\overline{p_u}\oplus   p_u\otimes  a_{(u,v)}\big) = \underset{\substack{
v: ~(u,v)\in E(C_n)}}{\bigoplus   }\big(\overline{p_u}\oplus   a_{(u,v)}\big),
\end{split}
\end{equation*}
where the last equality follows from the observation that i) when $p_u=0$, both $\overline{p_u}\oplus p_u\otimes a_{(u,v)}$ and $\overline{p_u}\oplus a_{(u, v)}$ are $a_{(u,v)}$, and ii) when $p_u=\infty$, both $\overline{p_u}\oplus p_u\otimes a_{(u,v)}=0$ and $\overline{p_u}\oplus a_{(u, v)}= 0\oplus a_{(u,v)} = 0$ (for both $a_{(u,v)}=0$ and $a_{(u,v)}=\infty$). So, as the format $Q(f):= \begin{pmatrix}
f & 0 \\
0 & \infty
\end{pmatrix}$ supports addition (as $Q(f)\otimes Q(\infty)\otimes Q(g) = Q(f\oplus g)$), we get an $n^{O(1)}$ sized width-2 ABP that computes $E_u$ in $Q$ format.

\subsection{$\mathsf{VNP}_{\mathsf{R}^{+}}$ via Hypercube sum over general Width-1 ABPs}
\label{hypercube sum over general width-1 ABPs section}

Over $\mathsf{R}^{+}$, we show that each of $B_{(u,v)}$'s, $C_u$'s, $D_u$'s, $E_u$'s and $A_{u}$'s (as in the previous subsection) are also expressible as hypercube sum over product of constant many linear forms. Then, multiplying all these expressions together would give a way to express $g_n$ (and so, also $f_n$) as a hypercube sum over product of polynomially-many linear forms, thereby strengthening $\mathsf{VNP}_{\mathsf{R}^{+}} = \mathsf{VNF}_{\mathsf{R}^{+}}$ to $\mathsf{VNP}_{\mathsf{R}^{+}} = \mathsf{VNBP_1^g}_{~\mathsf{R}^{+}}$, where $\mathsf{VNBP^g_1}_{~\mathsf{R^{+}}}$ consists of families expressible as hypercube sum over a family 
with polynomial-sized width-1 general (i.e., edge labels can be linear forms) ABPs over $\mathsf{R}^{+}$. 

Note that $B_{(u,v)}:= \overline{a_{(u,v)}}\oplus a_{(u,v)}\otimes p_u\otimes p_v$ can be written as $\overline{a_{(u,v)}}\oplus p_u\otimes p_v$. This is because i) when $a_{(u,v)}=0$, both $B_{(u,v)}$ and $\overline{a_{(u,v)}}\oplus p_u\otimes p_v$ are $p_u\otimes p_v$, and ii) when $a_{(u,v)}=\infty$, $B_{(u,v)}=0$ and $\overline{a_{(u,v)}}\oplus p_u\otimes p_v = 0\oplus p_u\otimes p_v$, which is $0$ for all $\infty$-$0$ substitutions of $p_u$ and $p_v$. Therefore, $B_{(u,v)} = \overline{a_{(u,v)}}\oplus p_u\otimes p_v$. This, in turn, can be written as the following hypercube sum: $$\underset{\lambda_{(u,v)}\in \{\infty, 0\}}{\bigoplus}(\overline{a_{(u,v)}}\oplus \overline{\lambda_{(u,v)}})\otimes (p_u\oplus \lambda_{(u,v)})\otimes (p_v\oplus \lambda_{(u,v)}),$$ with $\lambda_{(u,v)}$ as a new hypercube variable, because the summand (i.e., product of three linear forms) is $\overline{a_{(u,v)}}$ and $p_u\otimes p_v$ when $\lambda_{(u,v)}=0$ and $\infty$ respectively. Since $C_u:= \overline{p_u}\oplus p_u\otimes a_{(\ell(u),u)}\otimes a_{(r(u), u)}$ has a form similar to $B_{(u,v)}$, a similar argument would show that $C_u$ can be written as the following hypercube sum (with $\gamma_u$ as a new hypercube variable): $$\underset{\gamma_{u}\in \{\infty, 0\}}{\bigoplus}\big(\overline{p_u}\oplus \overline{\gamma_u}\big)\otimes (a_{(\ell(u),u)}\oplus \gamma_u\big)\otimes \big(a_{(r(u),u)}\oplus \gamma_u\big).$$
Next, note that $D_u:= \overline{p_u}\oplus~ p_u\otimes (a_{(\ell(u), u)}\otimes \overline{a_{(r(u), u)}} \oplus a_{(r(u), u)}\otimes \overline{a_{(\ell(u), u)}})$ can be written as $\overline{p_u}\oplus~ (a_{(\ell(u), u)}\otimes \overline{a_{(r(u), u)}} \oplus a_{(r(u), u)}\otimes \overline{a_{(\ell(u), u)}})$. This, in turn, can be written as the following hypercube sum (with $\alpha_u$ and $\beta_u$ as new hypercube variables): 
$$\hspace{-1 cm}\underset{\alpha_u, \beta_u\in \{\infty, 0\}}{\bigoplus}
(\overline{p_u}\oplus \overline{\alpha_u})\otimes (a_{(\ell(u), u)}\oplus \alpha_u\oplus \beta_u)\otimes (\overline{a_{(r(u),u)}}\oplus \alpha_u\oplus \beta_u)\otimes (\overline{a_{(\ell(u),u)}}\oplus \alpha_u\oplus \overline{\beta_u})\otimes (a_{(r(u), u)}\oplus \alpha_u\oplus \overline{\beta_u}).$$
This is because the summand (i.e., product of five linear forms) is i) $\overline{p_u}$ when $\alpha_u=0$, $\beta_u=0$, ii) $\overline{p_u}$ when $\alpha_u=0$, $\beta_u=\infty$, iii) $a_{(\ell(u),u)}\otimes \overline{a_{(r(u),u)}}$ when $\alpha_u=\infty, \beta_u=\infty$ and iv) $a_{(r(u),u)}\otimes \overline{a_{(\ell(u),u)}}$ when $\alpha_u= \infty$, $\beta_u=0$. 

Next, $E_u:= \overline{p_u}\oplus p_u\otimes \underset{v:(u,v)\in E(C_n)}{\bigoplus}a_{(u,v)}$ can be written as $\overline{p_u}\oplus \underset{v:(u,v)\in E(C_n)}{\bigoplus}a_{(u,v)}$, which is already a linear form. Finally, note that $A_{u}:=\overline{p_{u}}\oplus \text{label}(u)\otimes p_{u}$ can also be written as   $\overline{p_{u}}\oplus \text{label}(u)$, which is already a linear form too. The last re-expression is correct as i) when $p_{u}=0$, both $A_{u}$ and  $\overline{p_{u}}\oplus \text{label}(u)$ are $\text{label}(u)$, and ii) when $p_{u}=\infty$, $A_{u} = 0$ and $\overline{p_{u}}\oplus \text{label}(u) = 0\oplus \text{label}(u)$, which is $0$ for all possible substitutions of $\text{label}(u)$ from $\mathsf{R}^{+}$.

\subsection{Hamiltionian Cycle family in $\mathsf{VNP}_{\mathsf{R} (\text{or }\mathsf{R}^{+})}$}

Let us show that the Hamiltonian cycle family $(HC_n)_{n\geq 1}$ belongs to $\mathsf{VNP}_{\mathsf{R} (\text{or }\mathsf{R}^{+})}$. Recall that $HC_n:= \underset{\sigma\in C_n}{\bigoplus}\underset{1\leq i\leq n}{\bigotimes}x_{i,\sigma(i)}$. Note that every permutation $\sigma$ of $[n]$ can be specified by a $\infty$-$0$ permutation matrix $Y$ of size $n\times n$, where for every $1\leq i,j\leq n$, the $(i,j)^{th}$ entry of $Y$ indicates whether $\sigma(i)=j$. Also, all cyclic permutations (i.e., consisting of a single cycle) $\sigma$ correspond to those permutation matrices $Y$ wherein for every $1\leq u, v\leq n$, $v$ is in $u$'s orbit (under repeated application of the permutation $\sigma)$; that is, $(u,v)^{th}$ entry of the $k^{th}$ power of the matrix $Y$ is $0$ for some $1\leq k< n$. So,
\begin{equation*}
\begin{split}
HC_n & = \underset{Y\in \{\infty, 0\}^{n\times n}}{\bigoplus} \Bigg(\big[\substack{Y \text{ $\in P_n$}}\big]\otimes \Big(\underset{1\leq u,v\leq n}{\bigotimes}~\underset{0\leq k <n}{\bigoplus}\big(\substack{ (u,v)^{th} \text{ entry of } Y^{k}}\big)\Big)\otimes \Big(\underset{1\leq i\leq n}{\bigotimes}\underset{1\leq t\leq n}{\bigoplus}x_{i,t}\otimes Y_{i,t}\Big)\Bigg).
\end{split}
\end{equation*}
where $[Y \in P_n]$ denotes the predicate of whether $Y$ is a permutation matrix with entries $\infty$ and $0$, where in every row (resp. column) exactly one $0$ appears and all other entries are $\infty$.
Observe that $(u,v)^{th}$ entry of $Y^k$ can be computed by a polynomial-sized ABP (and so, circuit). This is because $Y^{k}$ can be computed by an ABP of width $n$ consisting of $k+1$ layers wherein the edges between any two layers are labelled by entries of $Y$; that is, for any $1\leq i, j\leq n$, the edge from $i^{th}$ node of any layer to $j^{th}$ node of the next layer is labelled with $Y_{i,j}$. Also, we have
\begin{equation*}
\begin{split}
[Y \in P_n] & = 
[Y \text{ has }\geq \text{one 0 in each row}]\otimes \big[\substack{\text{No two entries in same}\\ \text{row/column are both 0}}\big]\\
& \Big(\underset{1\leq i\leq n}{\bigotimes}\underset{1\leq j\leq n}{\bigoplus}Y_{i,j}\Big)\otimes \Bigg(\underset{\substack{(a,b), (c,d)\in [n]\times [n]:\\
a=c \text{ or } b=d,\\
\text{and } (a,b)\neq (c,d)}}{\bigotimes}(\overline{Y_{a,b}}\oplus \overline{Y_{c,d}})\Bigg),
\end{split}
\end{equation*}
which has a polynomial-sized circuit too. Hence, we get $(HC_n)_{n\geq 1}\in \mathsf{VNP}_{\mathsf{R} (\text{or } \mathsf{R}^{+})}$, as desired. 

In 2015, Grochow proved that over any totally ordered semiring, Hamiltonian cycle family cannot be obtained as a monotone affine $p$-projection of the permanent family (see Theorem 4.2 in \cite{grochow2017monotone}). So, in particular, $\mathsf{R}^{+}$, $(HC_n)_{n\geq 1}$ cannot be obtained as a $p$-projection of $(perm_n)_{n\geq 1}$ over $\mathsf{R}^{+}$. Therefore, since $(HC_{n})_{n\geq 1}\in \mathsf{VNP}_{\mathsf{R}^{+}}$, it follows that $(perm_n)_{n\geq 1}$ is not $\mathsf{VNP}_{\mathsf{R}^{+}}$ hard under $p$-projections.

\section{Non-Universality of Width-2 ABPs over $\mathsf{R}$ and $\mathsf{R}^{+}$}
\label{non-universality N section}

\nonuniversalwidthtwosemiringN*


\begin{proof}
For the sake of contradiction, assume that there is a width-$2$ ABP $\Gamma$ that computes $\bigoplus_{i=1}^3 (x_i\otimes y_i)$ over $\mathsf{R}$ (or $\mathsf{R}^{+}$). Assume that no edge of $\Gamma$ is labelled $\infty$. This is safe because any source-to-sink path containing an edge labelled $\infty$ contributes nothing (i.e., contributes $\infty$, which is the additive identity) to the sum of weights of all source-to-sink paths. First, we prove the following lemma: 
\begin{lem}
\label{Exactly one x1 and one y1}
For each $1\leq i\leq 3$, there is a source-to-sink path $P_i$ in $\Gamma$ that has exactly one $x_i$, exactly one $y_i$ and none of the other four variables (proved for $i=1$ below, and similar argument holds for $i=2$ and $i=3$ too).
\end{lem}

\begin{proof}
For the sake of contradiction, assume that no such path $P_1$ exists. Then, every source-to-sink path in $\Gamma$ containing both $x_1$ and $y_1$ is of one of the following three types: 1) $x_1$ appears at least twice, 2) $y_1$ appears at least twice, and 3) $x_1$ and $y_1$ appear exactly once each, but are accompanied by at least one of the other four variables. Note that $\bigoplus_{i=1}^3 (x_i\otimes y_i)$ is equal to the sum of weights of paths of these three types, along with the weights of paths that do not contain at least one of $x_1$ and $y_1$. Now, substitute $x_2=y_2=x_3=y_3=\infty$. Then, $\bigoplus_{i=1}^3(x_i\otimes y_i)$ becomes $x_1\otimes y_1$. The weights of paths that contain neither $x_1$ nor $y_1$ become constants; let $c\in \mathsf{R} (\text{or } \mathsf{R}^{+})$ be the sum of all these constants. The sum of weights of paths that contain $x_1$ but not $y_1$ takes the form $(a_1\otimes x_1)\oplus (a_2\otimes x_1^{\otimes 2})\oplus \ldots (a_t\otimes x_1^{\otimes t})$ for some $t\geq 0$ and $a_1, \ldots, a_t\in \mathsf{R} (\text{or } \mathsf{R}^{+})$. Similarly, the sum of weights of paths that contain $y_1$ but not $x_1$ takes the form $(b_1\otimes y_1)\oplus (b_2\otimes y_1^{\otimes 2})\oplus \ldots\oplus (b_r\otimes y_1^{\otimes r})$ for some $r\geq 0$ and $b_1, \ldots, b_r\in \mathsf{R} (\text{or } \mathsf{R}^{+})$. The weights of all Type 3 paths vanish (i.e., become $\infty$). The sum of weights of all Type $1$ paths takes the form $x_1^{\otimes 2}\otimes \big(\bigoplus_{i=0}^{\mu-1} x_1^{\otimes i}\otimes p_{i}(y_1)\big)$ for some $\mu \geq 0$ and some polynomials $p_0, \ldots, p_{\mu-1}$ in $\mathsf{R}[y_1]$ (or $\mathsf{R}^{+}[y_1]$). Similarly, the sum of weights of all Type 2 paths takes the form $y_1^{\otimes 2}\otimes \big(\bigoplus_{j=0}^{\nu-1}y_1^{\otimes j}\otimes q_{j}(x_1)\big)$ for some $\nu \geq 0$ and some polynomials $q_0, \ldots, q_{\nu-1}$ in $\mathsf{R}[x_1]$ (or $\mathsf{R}^{+}[x_1]$). Thus, overall, $x_1\otimes y_1 = c \oplus  (a_1\otimes x_1\oplus a_2\otimes x_1^{\otimes 2}\oplus \ldots a_t\otimes x_1^{\otimes t})\oplus (b_1\otimes y_1\oplus b_2\otimes y_1^{\otimes 2}\oplus \ldots \oplus b_r\otimes y_1^{\otimes r}) \oplus \Big(x_1^{\otimes 2}\otimes \big( \bigoplus_{i=0}^{\mu-1}x_1^{\otimes i}\otimes p_{i}(y_1)\big)\Big) \oplus \Big(y_1^{\otimes 2}\otimes \big(\bigoplus_{j=0}^{\nu-1}y_1^{\otimes j}\otimes q_{j}(x_1)\big)\Big)$. Substituting $x_1=y_1=\infty$ both sides, we get $c=\infty$. Substituting $x_1=0$ and $y_1=\infty$ both sides, we get $\infty =(a_1\oplus a_2\oplus \ldots a_t)\oplus \bigoplus_{i=0}^{\mu-1} p_{i}(\infty)$. So, it follows that $a_1=a_2=\ldots =a_t=\infty$ and each of the $\mu$ polynomials $p_0(y_1), \ldots, p_{\mu-1}(y_1)$ has no constant term (i.e., it is $\infty$). Similarly, substituting $x_1=\infty$ and $y_1=0$, we get $\infty = (b_1\oplus b_2\oplus \ldots \oplus b_r)\oplus \bigoplus_{j=0}^{\nu-1}q_j(\infty)$. So, it follows that $b_1=\ldots = b_r=\infty$ and each of the $\nu$ polynomials $q_0(x_1), \ldots, q_{\nu-1}(x_1)$ has no constant term. Let $p_{0}',\ldots, p_{\mu-1}'$ denote the polynomials obtained by pulling out a factor of $y_1$ from $p_0, \ldots, p_{\mu-1}$ respectively. Similarly, let $q_{0}', \ldots, q_{\nu-1}'$ denote the polynomials obtained by pulling out a factor of $x_1$ from $q_0, \ldots, q_{\nu-1}$ respectively. Thus, $x_1\otimes y_1= \Big(x_1^{\otimes 2}\otimes y_1\otimes \big(\bigoplus_{i=0}^{\mu-1}x_1^{i}\otimes p_i'(y_1)\big)\Big)$ $\oplus$ $\Big(x_1\otimes y_1^{\otimes 2}\otimes \big(\bigoplus_{j=0}^{\nu-1}y_1^{j}\otimes q_j'(x_1)\big)\Big)$. Pull out a factor of $x_1\otimes y_1$ (same as LHS) from RHS; what remains must be $0$ for all $x_1,y_1\in \mathbb{R}$ (or $\mathbb{R}_{\geq 0}$) as LHS = RHS. That is, for all $x_1, y_1\in \mathbb{R}$ (or $\mathbb{R}_{\geq 0}$),
we have $\Big(x_1\otimes \big(\bigoplus_{i=0}^{\mu-1}x_1^{i}\otimes p_i'(y_1)\big)\Big)\oplus \Big(y_1\otimes \big(\bigoplus_{j=0}^{\nu-1}y_1^{j}\otimes q_j'(x_1)\big)\Big) = 0$.
However, any monomial of non-$\infty$ coefficient in LHS contains at least one of $x_1$ and $y_1$; so, if the values substituted for $x_1$ and $y_1$ are strictly increased, the value taken by any monomial of non-$\infty$ coefficient in the LHS also strictly increases, thereby making it impossible for LHS to remain $0$ (i.e., RHS), a contradiction. This proves Lemma \ref{Exactly one x1 and one y1}. 
\end{proof}

Amongst the edges that appear in $E(P_1)\cup E(P_2)\cup E(P_3)$, consider a first (moving from left to right) edge that is labelled by an indeterminate. Without loss of generality, assume that this edge belongs to path $P_1$, and it is labelled by $x_1$. Also, let $i$ and $i+1$ denote the indices of the layers containing the tail and head of this edge respectively.
This edge has four possible orientations, i.e., top level of layer $i$ to top level of layer $i+1$, top level of layer $i$ to bottom level of layer $i+1$, bottom level of layer $i$ to top level of layer $i+1$, and bottom level of layer $i$ to bottom level of layer $i+1$ (see Figure \ref{non-universal 1}). Next, amongst the edges that appear in $E(P_2)\cup E(P_3)$ after layer $i$, consider a first (again, moving from left to right) edge that is labelled by an indeterminate. Without loss of generality, assume that this edge belongs to path $P_2$, and it is labelled by $x_2$. Also, let $j$ and $j+1$ denote the indices of the layers containing the tail and head of this edge respectively.  Assume that this edge is directed from top level of layer $j$ to top level of layer $j+1$ (see Figure \ref{non-universal 2}); the cases corresponding to the other three orientations of this edge can be analyzed in a similar way. 

\begin{figure}[ht!]
\centering
\includegraphics[scale=0.65]{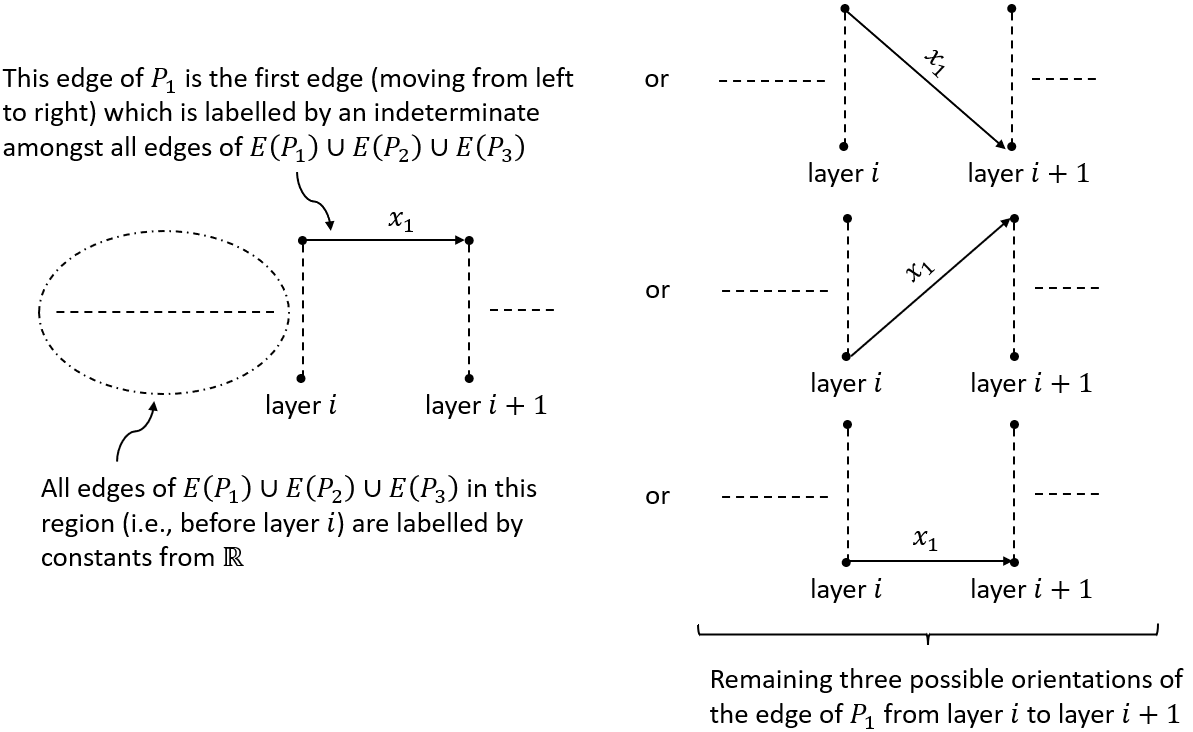}
\caption{}
\label{non-universal 1}
\end{figure}

\begin{figure}[ht!]
\centering
\includegraphics[scale=0.65]{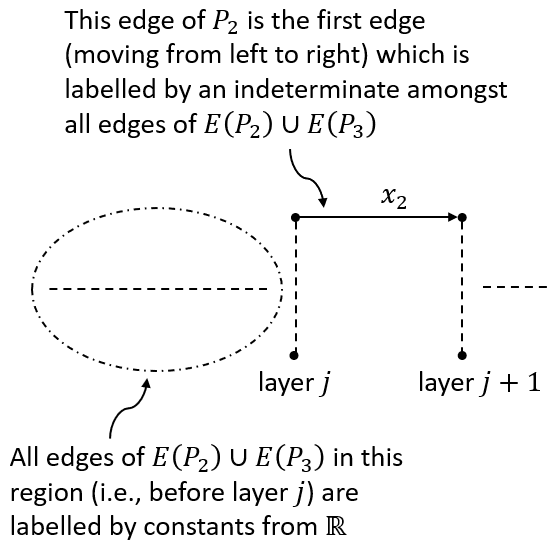}
\caption{ }
\label{non-universal 2}
\end{figure}

Consider the following two cases: 1) $j\neq i$ and 2) $j=i$. We present the proof for Case 1 now, and we defer the proof of Case 2 to Appendix \ref{Case 2 non-universality N appendix} as its analysis is similar to that of Case 1. First, we show that the edge of $P_1$ from layer $j-1$ to layer $j$ must have its head at the bottom level of layer $j$. Suppose not. That is, assume that $P_1$'s edge from layer $j-1$ to layer $j$ has its head at top level (see Figure \ref{non-universality N conc paths img 1}). Then, $P_2$'s portion from layers $\leq j$ has no variable, and $P_1$'s portion from layers $\geq j$ has either zero or one $y_1$. So, concatenating $P_2$'s portion from layers $\leq j$ with $P_1$'s portion from layers $\geq j$ gives a source-to-sink path in $\Gamma$ whose weight is of the form $c$ or $c\otimes y_1$ for some $c\in \mathbb{R}$ (or $\mathbb{R}_{\geq 0}$). We derive a contradiction in both cases as follows: In the former case, $(x_1\otimes y_1)\oplus (x_2\otimes y_2)\oplus (x_3\otimes y_3)\leq c$. Substituting $x_1=x_2=x_3=c+1$ and $y_1=y_2=y_3=0$ gives  $c+1\leq c$. In the latter case, $(x_1\otimes y_1)\oplus (x_2\otimes y_2)\oplus (x_3\otimes y_3) \leq c\otimes y_1$. Substituting $x_1=x_2=x_3=c+1$ and $y_1=y_2=y_3=0$ gives  $c+1\leq c$. This proves that the edge of $P_1$ from layer $j-1$ to layer $j$ has its head at bottom level of layer $j$.



\begin{figure}[ht!]
\includegraphics[scale=0.65]{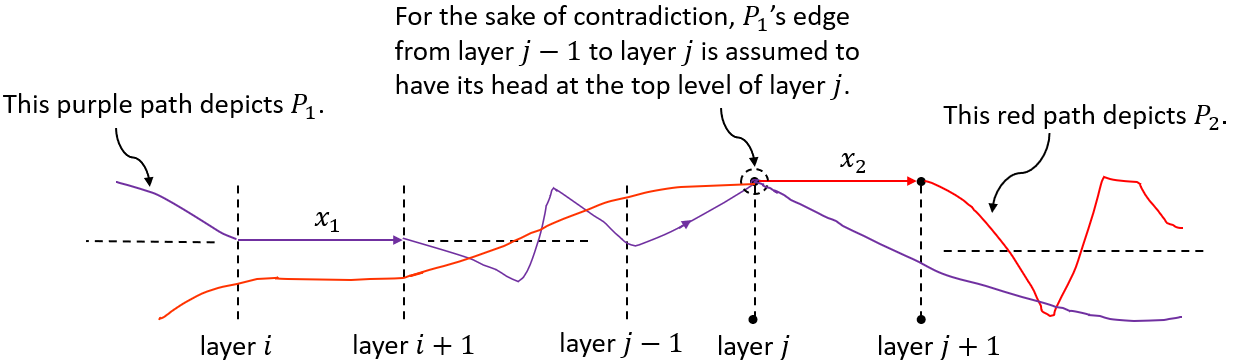}
\caption{}
\label{non-universality N conc paths img 1}
\end{figure}

Next, consider the edge of $P_1$ from layer $j$ to layer $j+1$. Its tail is same as the head of the edge of $P_1$ from layer $j-1$ to layer $j$ (which, as argued above, is at the bottom level of layer $j$). We show that the head of this edge must be at the bottom level of layer $j+1$. Suppose not. That is, assume that the edge of $P_1$ from layer $j$ to layer $j+1$ is directed from bottom level to top level (see Figure \ref{non-universality N conc paths img 2}). Then, $P_2$'s portion from layers $\leq j+1$ has one $x_2$, and $P_1$'s portion from layers $\geq j+1$ has either zero or one $y_1$. So, concatenating $P_2$'s portion from layers $\leq j+1$ with $P_1$'s portion from layers $\geq j+1$ gives a source-to-sink path in $\Gamma$ whose weight is of the form $c\otimes x_2$ or $c\otimes x_2\otimes y_1$ for some $c\in \mathbb{R}$ (or $\mathbb{R}_{\geq 0}$). We derive a contradiction in both cases as follows: In the former case, $(x_1\otimes y_1)\oplus (x_2\otimes y_2)\oplus (x_3\otimes y_3) \leq c \otimes x_2$. Substituting $x_1=x_2=x_3=0$ and $y_1=y_2=y_3=c+1$ gives $c+1\leq c$. In the latter case, $(x_1\otimes y_1)\oplus (x_2\otimes y_2)\oplus (x_3\otimes y_3)\leq c\otimes x_2\otimes y_1$. Substituting $y_1=x_2=y_3=0$ and $y_2=x_1=x_3= c+1$ gives $c+1\leq c$. This proves that the edge of $P_1$ from layer $j$ to layer $j+1$ has its head at bottom level of layer $j+1$.

\begin{figure}[ht!]
\centering
\includegraphics[scale=0.65]{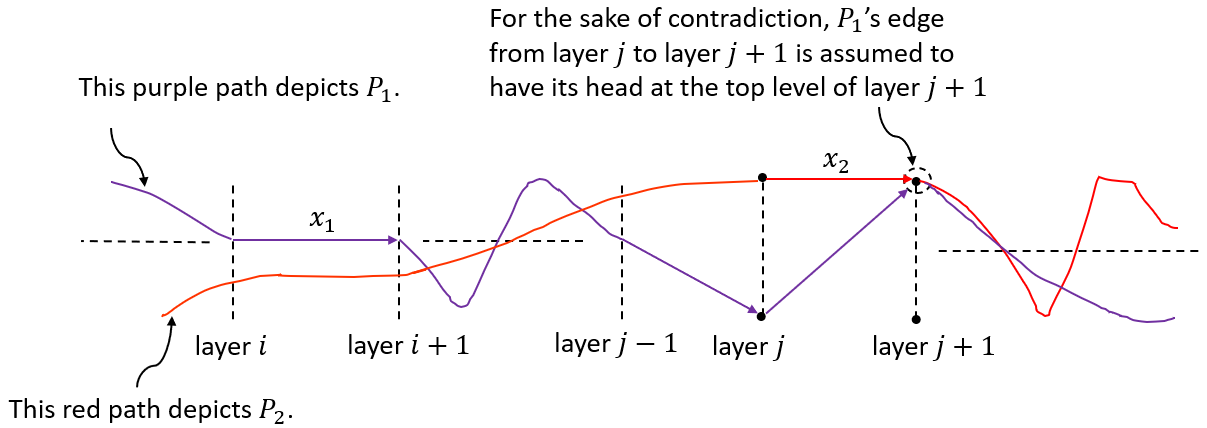}
\caption{}
\label{non-universality N conc paths img 2}
\end{figure}

Now, consider the edge of $P_3$ from layer $j$ to layer $j+1$. This edge must be different from the edge joining the top level of layer $j$ to the top level of layer $j+1$; this is because the latter edge is labelled with $x_2$, which does not belong to $P_3$. Also, we show that this edge cannot be directed from bottom level of layer $j$ to top level of layer $j+1$. Suppose not (see Figure \ref{non-universality N conc paths img 3}). Then, $P_2$'s portion from layers $\geq j+1$ has one $y_2$, and $P_3$'s portion from layers $\leq j+1$ has no variables or one $x_3$ or one $y_3$. So, concatenating $P_3$'s portion from layers $\leq j+1$ with $P_2
$'s portion from layers $\geq j+1$ gives a source-to-sink path in $\Gamma$ whose weight is of the form $c\otimes y_2$ or $c\otimes y_2\otimes x_3$ or $c\otimes y_2 \otimes y_3$ for some $c\in \mathbb{R}$ (or $\mathbb{R}_{\geq 0}$). We derive a contradiction in these three cases as follows: In the first case, $(x_1\otimes y_1)\oplus (x_2\otimes y_2)\oplus (x_3\otimes y_3)\leq c\otimes y_2$. Substituting $x_1=x_2=x_3=c+1$ and $y_1=y_2=y_3=0$ gives $c+1 \leq c$. In the second case, $(x_1\otimes y_1)\oplus (x_2\otimes y_2)\oplus (x_3\otimes y_3) \leq c\otimes y_2\otimes x_3$. Substituting $x_1=y_2=x_3=0$ and $x_2=y_1=y_3=c+1$ gives $c+1\leq c$. In the third case, $(x_1\otimes y_1)\oplus (x_2\otimes y_2)\oplus (x_3\otimes y_3) \leq c\otimes y_2\otimes y_3$. Substituting $y_1=y_2=y_3=0$ and $x_1=x_2=x_3=c+1$ gives $c+1\leq c$. 

\begin{figure}[ht!]
\centering
\includegraphics[scale=0.65]{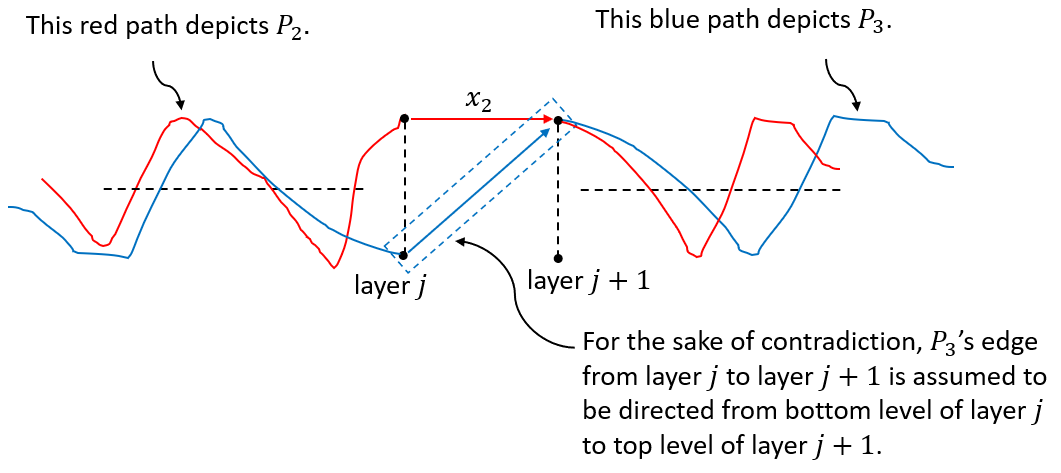}
\caption{}
\label{non-universality N conc paths img 3}
\end{figure}

Thus, the edge of $P_3$ from layer $j$ to layer $j+1$ has its head at the bottom level of layer $j+1$. Also, as proved earlier, the edge of $P_1$ from layer $j$ to layer $j+1$ has its head at the bottom level of layer $j+1$. Now, see Figure \ref{non-universality N conc paths img 4}. Note that $P_1$'s portion from layers $\geq j+1$ has zero or one $y_1$. Also, $P_3$'s portion from layers $\leq j+1$ has i) no variables, or ii) one $x_3$, or iii) one $y_3$. So, concatenating $P_1$'s portion from layers $\geq j+1$ with $P_3$'s portion from layers $\leq j+1$ gives a source-to-sink path in $\Gamma$ whose weight is of the form $c$,  $c\otimes x_3$, $c\otimes y_3$, $c\otimes y_1$, $c\otimes y_1\otimes x_3$ or $c\otimes y_1\otimes y_3$ for some $c\in \mathbb{R}$ (or $\mathbb{R}_{\geq 0}$). We derive a contradiction in these six cases as follows: In the first case, $(x_1\otimes y_1)\oplus (x_2\otimes y_2)\oplus (x_3\otimes y_3) \leq c$. Substituting $x_1=x_2=x_3=0$ and $y_1=y_2=y_3=c+1$ gives $c+1\leq c$. In the second case, $(x_1\otimes y_1)\oplus (x_2\otimes y_2)\oplus (x_3\otimes y_3) \leq c\otimes x_3$. Substituting $x_1=x_2=x_3=0$ and $y_1=y_2=y_3=c+1$ gives $c+1\leq c$. In the third case, $(x_1\otimes y_1)\oplus (x_2\otimes y_2)\oplus (x_3\otimes y_3) \leq c\otimes y_3$. Substituting $x_1=x_2=x_3=c+1$ and $y_1=y_2=y_3=0$ gives $c+1 \leq c$. In the fourth case, $(x_1\otimes y_1)\oplus (x_2\otimes y_2)\oplus (x_3\otimes y_3)\leq c\otimes y_1$. Substituting $x_1=x_2=x_3=c+1$ and $y_1=y_2=y_3=0$ gives $c+1\leq c$. In the fifth case, $(x_1\otimes y_1)\oplus (x_2\otimes y_2)\oplus (x_3\otimes y_3)\leq c\otimes y_1\otimes x_3$. Substituting $y_1=x_2=x_3=0$ and $x_1=y_2=y_3=c+1$ gives $c+1\leq c$. In the sixth case, $(x_1\otimes y_1)\oplus (x_2\otimes y_2)\oplus (x_3\otimes y_3) \leq c\otimes y_1\otimes y_3$. Substituting $y_1=y_2=y_3=0$ \& $x_1=x_2=x_3=c+1$ gives $c+1\leq c$. This proves Theorem \ref{non-universal width-2 semiring N}.
\begin{figure}[ht!]
\centering
\includegraphics[scale=0.65]{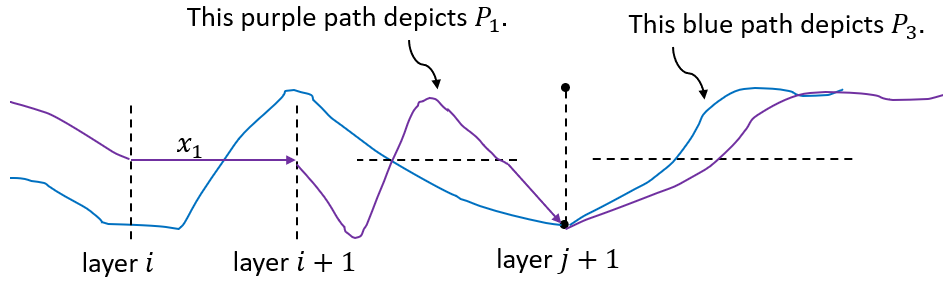}
\caption{}
\label{non-universality N conc paths img 4}
\end{figure}

\end{proof}

\section{Low Depth Formulas to Low Width ABPs over $\mathsf{R}$ and $\mathsf{R}^{+}$}
\label{low depth circuits to low width ABPs section}




\lowdepthformulastolowwidthABPs*

So, in particular, logarithmic depth alternating $\bigoplus\bigotimes \ldots \bigoplus \bigotimes$ formulas can be efficiently simulated using logarithmic width ABPs. We present the proof of Theorem \ref{Low depth circuits to low width ABPs} below.


\begin{proof}
Moving bottom-to-top in the formula $C$, we number the levels as $1, \ldots, 2p$, and let $s_1, \ldots, s_{2p}$
denote the fan-in's of gates in these levels respectively (see Figure \ref{sigma pi circuit figure}). For each $1\leq i\leq 2p$, define the format matrix $M_i(f)$ (associated with Level $i$) as follows\footnote{The non-$\infty$ values are highlighted in red for clarity.}: 

$$M_i(f) := ~~~\begin{pNiceMatrix}[first-col, first-row]
&  & \underset{\downarrow}{\mbox{Col}_2} & & & \underset{\downarrow}{\mbox{Col}_{2p+2-i}} &  \\
& \textcolor{red}{f} & \infty & \ldots & \infty & \textcolor{red}{0} & \infty & \ldots & \infty \\
& \infty &  & \adots & \textcolor{red}{0} & \adots & & & \vdots \\
& \vdots & \adots & \textcolor{red}{\adots} &   \adots &  & &  & \vdots \\
\mbox{Row}_{2p+1-i} \rightarrow &  \infty & \textcolor{red}{0} & \infty & \ldots & \ldots & \ldots & \ldots & \infty\\
&  \infty & \ldots &\ldots & \ldots & \ldots & \ldots & \ldots & \infty
\\
& \vdots &  & &  &  & & & \vdots\\
& \infty & \ldots & \ldots & \ldots & \ldots & \ldots & \ldots & \infty 
\end{pNiceMatrix}_{(2p+1)\times (2p+1)}$$

    
\begin{figure}[ht!]
\centering
\includegraphics[scale=0.7, trim = {0 0 0 3}, clip=true]{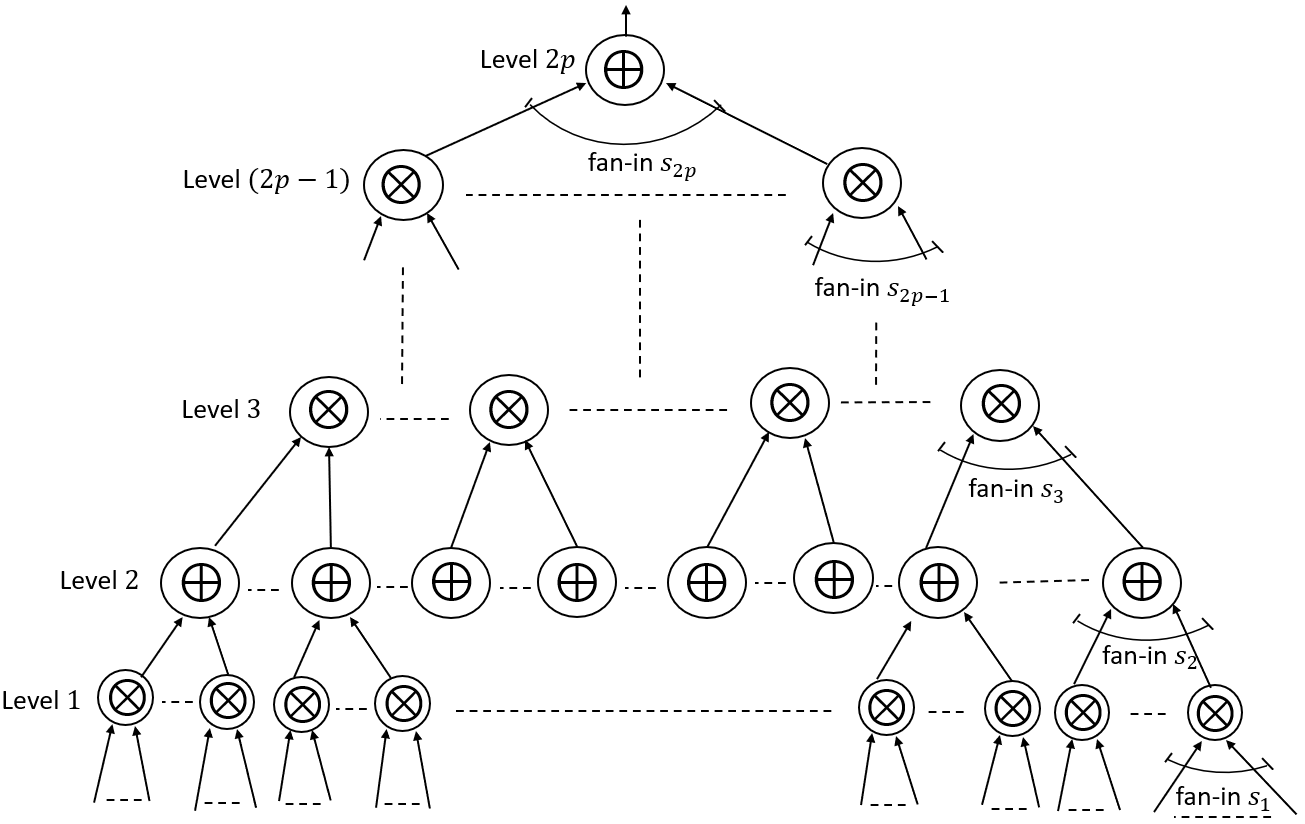}
\caption{The depth $2p$ alternating $\bigoplus\bigotimes\ldots \bigoplus\bigotimes$ formula $C$ in Theorem \ref{Low depth circuits to low width ABPs}. The levels are $1, 2, \ldots, 2p$ (moving bottom-to-top), and fan-in's of gates in these levels are $s_1, \ldots, s_{2p}$ respectively.}
\label{sigma pi circuit figure}
\end{figure}


Observe the following format conversions:
\begin{itemize}
\item Using $M_1(f)$ once to get $M_1(f\otimes x)$ (for any variable/constant $x$):
\begin{equation*}
\underbrace{\begin{pmatrix}
\textcolor{red}{f} & \infty & \ldots & \infty & \textcolor{red}{0} \\
\infty &  & \adots & \textcolor{red}{0} & \infty  \\
\vdots & \adots & \textcolor{red}{\adots} &   \adots & \vdots \\
  \infty & \textcolor{red}{0} & \infty & \ldots & \infty\\
\infty & \ldots &\ldots & \ldots & \infty
\end{pmatrix}}_{M_1(f)}\otimes 
\begin{pmatrix}
\textcolor{red}{x} & \infty & \ldots & \ldots & \infty \\
\infty & \textcolor{red}{0}  &  &  & \vdots  \\
\vdots &  & \textcolor{red}{\ddots} &   & \vdots \\
 \vdots &  &  & \textcolor{red}{0} & \infty\\
\infty & \ldots &\ldots & \ldots & \textcolor{red}{0}
\end{pmatrix} = \underbrace{\begin{pmatrix}
\textcolor{red}{f\otimes x} & \infty & \ldots & \infty & \textcolor{red}{0} \\
\infty &  & \adots & \textcolor{red}{0} & \infty  \\
\vdots & \adots & \textcolor{red}{\adots} &   \adots & \vdots \\
  \infty & \textcolor{red}{0} & \infty & \ldots & \infty\\
\infty & \ldots &\ldots & \ldots & \infty
\end{pmatrix}}_{M_1(f\otimes x)}
\end{equation*}

\item For every $2\leq i\leq 2p$,
\begin{itemize}
\item Using $M_i(f)$ and $M_{i-1}(g)$ once each to get $M_i(f\otimes g)$:
\begin{equation*}
\begin{split}
\begin{pNiceMatrix}[first-col, first-row]
&  & \underset{\downarrow}{\mbox{Col}_2} & & & \underset{\downarrow}{\mbox{Col}_{2p+2-i}} &  \\
& \textcolor{red}{f} & \textcolor{red}{0} & \infty & \ldots & \ldots & \ldots & \ldots & \infty \\
& \infty & \infty & \textcolor{red}{0} & \ddots &  & & & \infty \\
& \vdots & & \ddots & \textcolor{red}{\ddots} &   \ddots &   &  & \vdots \\
\mbox{Row}_{2p+1-i} \rightarrow &  \infty & \ldots & \ldots & \infty & \textcolor{red}{0} & \infty & \ldots & \infty\\
&  \infty & \ldots &\ldots & \ldots & \ldots & \ldots & \ldots & \infty
\\
& \infty & \ldots & \ldots & \ldots & \ldots & \ldots & \ldots & \infty
\\
& \vdots &  & &  &  & & & \vdots\\
& \infty & \ldots & \ldots & \ldots & \ldots & \ldots & \ldots & \infty 
\end{pNiceMatrix}\otimes \\
\begin{pNiceMatrix}[first-col, first-row]
&  & \underset{\downarrow}{\mbox{Col}_2} & & & \underset{\downarrow}{\mbox{Col}_{2p+2-i}} &  \\
& \textcolor{red}{g} & \infty & \ldots & \ldots & \infty & \infty & \ldots & \infty \\
& \infty &  &  &  & \textcolor{red}{0} & \infty & & \infty \\
& \vdots & &  & \textcolor{red}{\adots} & \adots   &   &  & \vdots \\
 &  \infty &  & \textcolor{red}{\adots} & \adots &  & &  & \vdots\\
&  \infty & \textcolor{red}{0} &\adots &  &  &  &  & \vdots
\\
\mbox{Row}_{2p+3-i} \rightarrow & \textcolor{red}{0} & \infty & &  &  & & & \vdots\\
& \vdots &  & \ &  &  &  &  & \vdots
\\
& \infty & \ldots & \ldots & \ldots & \ldots & \ldots & \ldots & \infty 
\end{pNiceMatrix}\\
=~~~~ \begin{pNiceMatrix}[first-col, first-row]
&  & \underset{\downarrow}{\mbox{Col}_2} & & & \underset{\downarrow}{\mbox{Col}_{2p+2-i}} &  \\
& \textcolor{red}{f\otimes g} & \infty & \ldots & \infty & \textcolor{red}{0} & \infty & \ldots & \infty \\
& \infty &  & \adots & \textcolor{red}{0} & \adots & & & \vdots \\
& \vdots & \adots & \textcolor{red}{\adots} &   \adots &  & &  & \vdots \\
\mbox{Row}_{2p+1-i} \rightarrow &  \infty & \textcolor{red}{0} & \infty & \ldots & \ldots & \ldots & \ldots & \infty\\
&  \infty & \ldots &\ldots & \ldots & \ldots & \ldots & \ldots & \infty
\\
& \vdots &  & &  &  & & & \vdots\\
& \infty & \ldots & \ldots & \ldots & \ldots & \ldots & \ldots & \infty 
\end{pNiceMatrix},
\end{split}
\end{equation*}
where first matrix in LHS is obtained by permuting columns $2, \ldots, 2p+2-i$ of $M_{i}(f)$, second matrx in LHS is transpose of $M_{i-1}(g)$, and the matrix in RHS is $M_{i}(f\otimes g)$.

\item Using $M_i(f)$ and $M_{i-1}(g)$ once each to get $M_i(f\oplus g)$:
\begin{equation*}
\begin{split}
\begin{pNiceMatrix}[first-col, first-row]
&  & \underset{\downarrow}{\mbox{Col}_2} & & & \underset{\downarrow}{\mbox{Col}_{2p+2-i}} &  \\
& \textcolor{red}{f} & \infty & \ldots & \infty & \textcolor{red}{0} & \infty & \ldots & \infty \\
& \infty &  & \adots & \textcolor{red}{0} & \adots & & & \vdots \\
& \vdots & \adots & \textcolor{red}{\adots} &   \adots &  & &  & \vdots \\
\mbox{Row}_{2p+1-i} \rightarrow &  \infty & \textcolor{red}{0} & \infty & \ldots & \ldots & \ldots & \ldots & \infty\\
&  \infty & \ldots &\ldots & \ldots & \ldots & \ldots & \ldots & \infty
\\
& \vdots &  & &  &  & & & \vdots\\
& \infty & \ldots & \ldots & \ldots & \ldots & \ldots & \ldots & \infty 
\end{pNiceMatrix}\otimes \\
\begin{pNiceMatrix}[first-col, first-row]
&  &  & & & \underset{\downarrow}{\mbox{Col}_{2p+2-i}} &  \\
& \textcolor{red}{0} & \infty & \ldots & \ldots & \ldots & \ldots & \ldots & \infty \\
& \infty & \textcolor{red}{0} & \infty & \ddots &  & & & \infty \\
& \vdots & & \textcolor{red}{\ddots} & \ddots &   \ddots &   &  & \vdots \\
 &  \infty & \ldots & \ldots & \textcolor{red}{0} & \infty & \infty & \ldots & \infty\\
\mbox{Row}_{2p+2-i} \rightarrow &  \textcolor{red}{g} & \ldots &\ldots & \ldots & \textcolor{red}{0} & \ldots & \ldots & \infty
\\
& \infty & \ldots & \ldots & \ldots & \ldots & \ldots & \ldots & \infty
\\
& \vdots &  & &  &  & & & \vdots\\
& \infty & \ldots & \ldots & \ldots & \ldots & \ldots & \ldots & \infty 
\end{pNiceMatrix}\\
= ~~~\begin{pNiceMatrix}[first-col, first-row]
&  & \underset{\downarrow}{\mbox{Col}_2} & & & \underset{\downarrow}{\mbox{Col}_{2p+2-i}} &  \\
& \textcolor{red}{f\oplus g} & \infty & \ldots & \infty & \textcolor{red}{0} & \infty & \ldots & \infty \\
& \infty &  & \adots & \textcolor{red}{0} & \adots & & & \vdots \\
& \vdots & \adots & \textcolor{red}{\adots} &   \adots &  & &  & \vdots \\
\mbox{Row}_{2p+1-i} \rightarrow &  \infty & \textcolor{red}{0} & \infty & \ldots & \ldots & \ldots & \ldots & \infty\\
&  \infty & \ldots &\ldots & \ldots & \ldots & \ldots & \ldots & \infty
\\
& \vdots &  & &  &  & & & \vdots\\
& \infty & \ldots & \ldots & \ldots & \ldots & \ldots & \ldots & \infty 
\end{pNiceMatrix},
\end{split}
\end{equation*}
where first matrix in LHS is $M_i(f)$, the matrix in RHS is $M_i(f\oplus g)$, and second matrix in LHS is obtained from $M_{i-1}(g)$ by changing Col$_{1}$ to Col$_{1} \oplus $Col$_{2}$, scaling Col$_2$ by $\infty$, permuting columns $2, \ldots, 2p+3-i$, and permuting rows $1, \ldots, 2p+2-i$ as follows:

\begin{equation*}
\begin{split}
& \begin{pNiceMatrix}[first-col, first-row]
&  & \underset{\downarrow}{\mbox{Col}_2} & & & \underset{\downarrow}{\mbox{Col}_{2p+3-i}} &  \\
& \textcolor{red}{g} & \infty & \ldots & \infty & \textcolor{red}{0} & \infty & \ldots & \infty \\
& \infty &  & \adots & \textcolor{red}{0} & \adots & & & \vdots \\
& \vdots & \adots & \textcolor{red}{\adots} &   \adots &  & &  & \vdots \\
\mbox{Row}_{2p+2-i} \rightarrow &  \infty & \textcolor{red}{0} & \infty & \ldots & \ldots & \ldots & \ldots & \infty\\
&  \infty & \ldots &\ldots & \ldots & \ldots & \ldots & \ldots & \infty
\\
& \vdots &  & &  &  & & & \vdots\\
& \infty & \ldots & \ldots & \ldots & \ldots & \ldots & \ldots & \infty 
\end{pNiceMatrix}\xrightarrow[\text{Scale } \text{Col}_2 \text{ by } \infty]{\text{Col}_1\leftarrow \text{Col}_1\oplus \text{Col}_2}\\
& \begin{pNiceMatrix}[first-col, first-row]
&  & \underset{\downarrow}{\mbox{Col}_2} & \underset{\downarrow}{\mbox{Col}_3} & & & \underset{\downarrow}{\mbox{Col}_{2p+3-i}} &  \\
& \textcolor{red}{g} & \infty & \infty & \ldots & \infty & \textcolor{red}{0} & \infty & \ldots & \infty \\
& \infty & &  & \adots & \textcolor{red}{\adots} & \adots & & & \vdots \\
& \vdots & & \adots & \textcolor{red}{\adots} &   \adots &  & &  & \vdots \\
 &  \infty & \infty & \textcolor{red}{0} & \infty & \ldots & \ldots & \ldots & \ldots & \infty\\
 \mbox{Row}_{2p+2-i} \rightarrow&  \textcolor{red}{0} & \ldots &\ldots & \ldots & \ldots & \ldots & \ldots & \ldots & \infty
\\
& \vdots &  & &  &  & & & & \vdots\\
& \infty & \ldots & \ldots & \ldots & \ldots & \ldots & \ldots & \ldots & \infty 
\end{pNiceMatrix}\xrightarrow[]{\text{Permute Col}_2, \ldots \text{Col}_{2p+3-i}}
\end{split}
\end{equation*}

\begin{equation*}
\begin{split}
& \begin{pNiceMatrix}[first-col, first-row]
&  &  & & & \underset{\downarrow}{\mbox{Col}_{2p+2-i}} &  \\
& \textcolor{red}{g} & \infty & \ldots & \infty & \textcolor{red}{0} & \infty & \ldots & \infty \\
& \infty &  & \adots & \textcolor{red}{0} & \adots & & & \vdots \\
& \vdots & \adots & \textcolor{red}{\adots} &   \adots &  & &  & \vdots \\
 &  \infty & \textcolor{red}{0} & \infty & \ldots & \ldots & \ldots & \ldots & \infty\\
\mbox{Row}_{2p+2-i} \rightarrow &  \textcolor{red}{0} & \infty &\ldots & \ldots & \ldots & \ldots & \ldots & \infty
\\
& \vdots &  & &  &  & & & \vdots\\
& \infty & \ldots & \ldots & \ldots & \ldots & \ldots & \ldots & \infty 
\end{pNiceMatrix}\xrightarrow[]{\text{Permute Row}_1, \ldots \text{Row}_{2p+2-i}}\\
& \begin{pNiceMatrix}[first-col, first-row]
&  &  & & & \underset{\downarrow}{\mbox{Col}_{2p+2-i}} &  \\
& \textcolor{red}{0} & \infty & \ldots & \ldots & \ldots & \ldots & \ldots & \infty \\
& \infty & \textcolor{red}{0} & \infty & \ddots &  & & & \infty \\
& \vdots & & \textcolor{red}{\ddots} & \ddots &   \ddots &   &  & \vdots \\
 &  \infty & \ldots & \ldots & \textcolor{red}{0} & \infty & \infty & \ldots & \infty\\
\mbox{Row}_{2p+2-i} \rightarrow &  \textcolor{red}{g} & \ldots &\ldots & \ldots & \textcolor{red}{0} & \ldots & \ldots & \infty
\\
& \infty & \ldots & \ldots & \ldots & \ldots & \ldots & \ldots & \infty
\\
& \vdots &  & &  &  & & & \vdots\\
& \infty & \ldots & \ldots & \ldots & \ldots & \ldots & \ldots & \infty 
\end{pNiceMatrix}.
\end{split}
\end{equation*}
\end{itemize}
\end{itemize}

Now, we describe how the above format conversions lead to the desired construction. First, for each Level 1 $\bigotimes$ gate, repeatedly use $M_1(f)$-to-$M_1(f\otimes x)$ conversion to get a sequence of $\mathcal{O}(s_1)$ $(2p+1)\times (2p+1)$ matrices that compute its output in $M_1$ format. Next, for each Level 2 $\bigoplus$ gate, repeatedly use $M_2(f)$-and-$M_1(g)$-to-$M_2(f\oplus g)$ conversion to get a sequence of $\mathcal{O}(s_1s_2)$ matrices that compute its output in $M_2$ format. Next, for each Level 3 $\bigotimes$ gate, repeatedly use $M_3(f)$-and-$M_2(g)$-to-$M_3(f\otimes g)$ conversion to get a seqeunce of $\mathcal{O}(s_1s_2s_3)$ matrices that computes its output in $M_3$ format. Continuing this process, we finally get a sequence of $\mathcal{O}(s_1s_2\ldots s_{2p}) = \mathcal{O}(s)$ matrices that compute the output of the top gate (i.e., Level $2p$ $\bigoplus$ gate) in $M_{2p}$ format.  To get ABP size, we put an extra multiplicative factor of $\mathcal{O}(p)$ as each matrix here has $\mathcal{O}(p)$ non-$\infty$ entries; so, there are $\mathcal{O}(p)$ edges between any two layers of the width $(2p+1)$ ABP so constructed. This proves Theorem \ref{Low depth circuits to low width ABPs}. 
\end{proof}

It is known that Shortest $s$-$t$ Path polynomial over $\mathsf{N}:=(\mathbb{N}\cup\{\infty\}, \oplus, \otimes)$ has a depth $2p$ alternating $\bigoplus\bigotimes\ldots \bigoplus\bigotimes$ formula of size $n^{\mathcal{O}(p\cdot n^{\frac{1}{p}})}$ (see Section 3.2 in \cite{mahajan2019shortest}). The corresponding pure $(\min, +)$ DP is as follows: For all vertices $i, j$ and lengths $\ell< n$, it stores $\mathbb{T}(i,j, \ell)$, i.e., minimum cost of any $i$-to-$j$ path that uses $\leq \ell$ edges. To compute $\mathbb{T}(i,j, \ell)$'s, it uses the following recurrence:
$$\mathbb{T}(i,j, \ell)= \underset{k_1, \ldots, k_{r-1}\in V(G)}{\min}\Big\{\mathbb{T}\Big(i, k_1, \frac{\ell}{r}\Big)+ \mathbb{T}\Big(k_1, k_2, \frac{\ell}{r}\Big) +\ldots+\mathbb{T}\Big(k_{r-2}, k_{r-1}, \frac{\ell}{r}\Big)+ \mathbb{T}\Big(k_{r-1}, j, \frac{\ell}{r}\Big)\Big\},$$
where $r:= n^{\frac{1}{p}}$. In the corresponding formula, every $\min$ (i.e., $\oplus$) gate has fan-in $\mathcal{O}(n^{r-1})$, every $+$ (i.e., $\otimes$) gate has fan-in $\mathcal{O}(r)$ and depth of the circuit is $2p$. So, the overall formula size is $n^{\mathcal{O}(p\cdot n^{\frac{1}{p}})}$.

This DP works because the optimal path can be obtained by guessing $r-1$ vertices at which the path would be broken if it were to be divided into $r$ equal length sub-paths, and then combining optimal paths of lengths $\leq \ell/r$ each between all pairs of consecutive guessed vertices (computed earlier by the DP). This may be problematic if these optimal paths share common vertices (and so, their merger gives a walk, instead of a path). It is not an issue over $\mathsf{N}$ (and also, $\mathsf{R}^{+}$) because when all costs are non-negative, an optimal path can be recovered from an optimal walk by skipping its portion between first and last appearances of any vertex (i.e., skipping cycles within the walk). Over $\mathsf{R}$, we cannot say the same as a skipped cycle may have negative cost (and so, its removal from the walk  increases the total cost); nevertheless, we can still use the same DP for acyclic graphs (as then, walks are same as paths) or more generally, graphs containing only non-negative cycles. 

Now, as the polynomial computed by any size $s$ ABP over $\mathsf{R}$ or $\mathsf{R}^{+}$ can be seen as a projection of Shortest Path polynomial of a size $s$ acyclic graph, we can use the above DP to simulate the ABP by a depth $2p$ alternating $\bigoplus\bigotimes\ldots \bigoplus\bigotimes$ formula of size $s^{\mathcal{O}(p\cdot s^{\frac{1}{p}})}$. Then, using Theorem \ref{Low depth circuits to low width ABPs}, this formula can be simulated by a width $(2p+1)$ ABP of size $s^{\mathcal{O}(p\cdot s^{\frac{1}{p}})}$. So, we get the following corollary:
\vspace{0.15 cm}

\ABPwidthreductioncorollary*


\section{Width-2 ABPs for $\infty$-$0$ Bivariate Polynomials over $\mathsf{R}^{+}$}
\label{bivariate construction section}

\bivariateconstructiontheorem*

\begin{proof}
Let $x^{\otimes a_1}\otimes y^{\otimes b_1}, x^{\otimes a_2}\otimes y^{\otimes b_2}, \ldots, x^{\otimes a_\ell}\otimes y^{\otimes b_\ell}$ denote the coefficient $0$ monomials of $f$, listed in non-increasing order of $x$ powers (i.e., $a_1\geq a_2 \geq \ldots \geq a_\ell$). Suppose that $a_i=a_j$ (say $=a$) for some $1\leq i<j\leq \ell$. Then, consider the $i^{th}$ and $j^{th}$ monomials, i.e., $x^{\otimes a}\otimes y^{\otimes b_i}$ and $x^{\otimes a}\otimes y^{\otimes b_j}$ respectively. If $b_i\geq b_j$, then $x^{\otimes a}\otimes y^{\otimes b_i} \geq x^{\otimes a}\otimes y^{\otimes b_j}$ for all possible substitutions of $x$ and $y$ from $\mathsf{R}^{+}$ and thus, the $i^{th}$ monomial can be safely dropped (i.e., it gets absorbed into the $j^{th}$ monomial). Similarly, if $b_j>b_i$, then $x^{\otimes a}\otimes y^{\otimes b_j}\geq x^{\otimes a}\otimes y^{\otimes b_i}$ for all possible substitutions of $x$ and $y$ from $\mathsf{R}^{+}$ and thus, the $j^{th}$ monomial can be safely dropped (i.e., it gets absorbed into the $i^{th}$ monomial). So, after exhaustively making such absorptions, we safely assume that $a_1>a_2>\ldots>a_\ell$. Next, let us analyze the relative order of magnitudes of $y$ powers. Suppose that $b_i\geq b_j$ for some $1\leq i<j\leq \ell$. Then, consider the $i^{th}$ and $j^{th}$ monomials, i.e., $x^{\otimes a_i}\otimes y^{\otimes b_i}$ and $x^{\otimes a_j}\otimes y^{\otimes b_j}$ respectively. Since $b_i\geq b_j$ and $a_i>a_j$, $x^{\otimes a_i}\otimes y^{\otimes b_i}\geq x^{\otimes a_j}\otimes y^{\otimes b_j}$ for all possible substitutions of $x$ and $y$ from $\mathsf{R}^{+}$ and thus, the $i^{th}$ monomial can be safely dropped (i.e., it gets absorbed into the $j^{th}$ monomial). So, after exhaustively making such absorptions, we safely assume that $b_1< b_2<\ldots <b_\ell$. We illustrate our construction for $\ell=3$ (i.e., three monomials) below, but it can be generalized to larger $\ell$'s (i.e., more monomials) too. For $\ell=3$, we have $f=(x^{\otimes a_1}\otimes y^{\otimes b_1})\oplus (x^{\otimes a_2}\otimes y^{\otimes b_2})\oplus (x^{\otimes a_3}\otimes y^{\otimes b_3})$. As argued above, $a_1>a_2>a_3$ and $b_1<b_2<b_3$. Observe that the width-2 ABP shown in Figure \ref{width-2 bivariate construction} computes $f$ over $\mathsf{R}^{+}$. This is because i) source-to-sink path passing through the leftmost bridge has weight $y^{\otimes (b_1-1)}\otimes y\otimes x^{\otimes (a_1-a_2)}\otimes x^{\otimes (a_2-a_3)}\otimes x^{\otimes a_3} = x^{\otimes a_1}\otimes y^{\otimes b_1}$, ii) source-to-sink path passing through the middle bridge has weight $y^{\otimes (b_1-1)}\otimes y\otimes y^{\otimes (b_2-b_1-1)}\otimes y\otimes x^{\otimes (a_2-a_3)}\otimes x^{\otimes a_3} = x^{\otimes a_2}\otimes y^{\otimes b_2}$, and iii) source-to-sink path passing through the rightmost bridge has weight $y^{\otimes (b_1-1)}\otimes y\otimes y^{\otimes (b_2-b_1-1)}\otimes y\otimes y^{\otimes (b_3-b_2-1)}\otimes y\otimes x^{\otimes a_3} = x^{\otimes a_3}\otimes y^{\otimes b_{3}}$. More generally, for any $\ell$, this gives a width-2 ABP of size $\mathcal{O}\big(b_1+\max\{a_1-a_2, b_2-b_1\}+\max\{a_2-a_3, b_3-b_2\}+\ldots +\max\{a_{\ell-1}-a_{\ell}, b_{\ell}-b_{\ell-1}\}+a_{\ell}\big)$, which is $\mathcal{O}(\operatorname{degree}(f))$. This proves Theorem \ref{uniform-coefficient bivariates over N}.
\begin{figure}[ht!]
\centering
\includegraphics[scale=0.6]{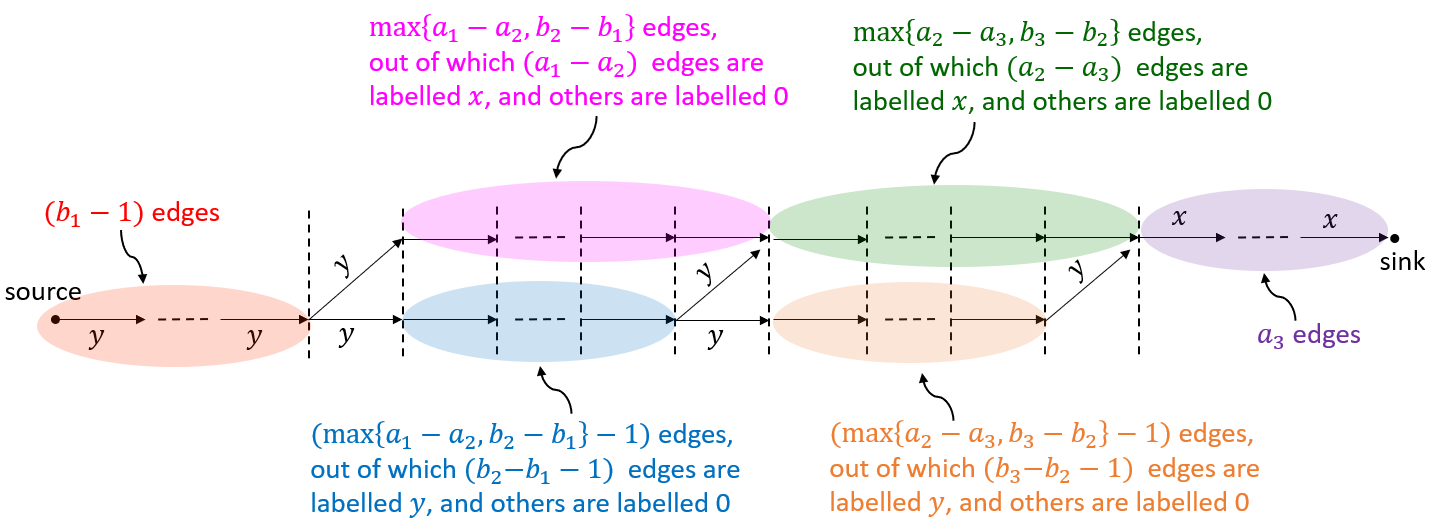}$~~~~~~~~~~~~~~~~~~~~~~~~~~~~~~~~~~~~~~~~~~~~~~~~~~~~~~~~~~~~~~~~~~~~~~~~~~~~~~$
\caption{: A width-2 ABP that computes the bivariate polynomial $(x^{\otimes a_1}\otimes y^{\otimes b_1})\oplus (x^{\otimes a_2}\otimes y^{\otimes b_2})\oplus (x^{\otimes a_3}\otimes y^{\otimes b_3})$ with three monomials (where $a_1>a_2>a_3$ and $b_1<b_2<b_3$) over $\mathsf{R}^{+}$.}
\label{width-2 bivariate construction}
\end{figure}    
\end{proof}

\section{Conclusion and Open Problems}
\label{sec: conclusion}
As mentioned earlier, the key question that motivated us to begin this work (but is still unresolved) is whether there is an analogue of Ben-Or \& Cleve's result \cite{cleve1988computing} over min-plus semirings $\mathsf{R}$ and $\mathsf{R}^{+}$. That is, can formulas be efficiently simulated using bounded-width ABPs over $\mathsf{R}$ and $\mathsf{R}^{+}$? Jukna observed that the sub-semiring $(\{\infty, 0\}, \oplus, \otimes)$ is isomorphic to the Boolean semiring $(\{0,1\}, \vee, \wedge)$ via the isomorphism $\infty \mapsto 0$ and $0\mapsto 1$ (see, for example, Lemma 7 in \cite{jukna2015lower}). Using his observation, it can be argued that an analogue of Ben-or \& Cleve's result over min-plus semirings would imply that polynomial-sized monotone Boolean formulas can be simulated using polynomial-sized monotone bounded-width branching programs, i.e., $m$-$\mathsf{NC}^{1} \subseteq m$-$\mathsf{BWBP}$ thus refuting a conjecture of  Grigni and Sipser \cite{grigni1992monotone}, that $\mathsf{Majority}$ function (which is known to belong to $m$-$\mathsf{NC^{1}}$) does not belong to $m$-$\mathsf{BWBP}$.

Another direction could be to build upon lower bound proof devised by Jukna and Schnitger in 2016 \cite{jukna2016optimality}. They showed that any ABP computing a polynomial $f$ (over $\mathsf{R}$ or $\mathsf{R}^{+}$) must have at least as many disjoint \emph{cuts}\footnote{Subset of variable labelled edges that intersects every source-to-sink path.} as the \emph{length}\footnote{Smallest monomial degree in any polynomial defining the same function as $f$.} of $f$; so, as size of each cut is at least the \emph{cover number}\footnote{Smallest number of variables which when substituted as $\infty$ make $f$ evaluate to $\infty$.} of $f$, the ABP's size must be at least the product of length and cover number. Their proof works even for unbounded-width ABPs, but it gives only polynomial lower bounds. Strengthening their argument in the context of bounded-width ABPs to get stronger lower bounds is an interesting open problem.

Next, although we defined an analogue of the class $\mathsf{VNP}$ over $\mathsf{R}$ and $\mathsf{R}^{+}$, it remains to identify $p$-families complete for the class (if any) under $p$-projections. As discussed earlier, permanent family $(perm_n)_{n\geq 1}$ cannot be $\mathsf{VNP}_{\mathsf{R}^{+}}$ complete due to \cite{grochow2017monotone}; but, is it $\mathsf{VNP}_{\mathsf{R}}$-complete? Another potential candidate polynomial for $\mathsf{VNP}$ over $\mathsf{R}$ and $\mathsf{R}^{+}$ is the Hamiltonian cycle family $(HC_n)_{n\geq 1}$. 

Also, our arguments for non-universality for width-2 ABPs (Theorem \ref{non-universal width-2 semiring N}) works only for weakest width-2 ABPs. It would be interesting to generalize it to even more general width-2 ABPs. Further, while we showed $\mathsf{VNP}_{\mathsf{R}^{+}} = \mathsf{VNBP_1^{g}}_{~\mathsf{R}^{+}}$, it is unclear if this hold true over $\mathsf{R}$ as well. If yes, it would subsume the result that $\mathsf{VNP}_{\mathsf{R}} = \mathsf{VNBP_2}_{~\mathsf{R}}$ as weakest width-2 ABPs can efficiently simulate general width-$1$ ABPs over $\mathsf{R}$. This is because for any constant $a\in \mathsf{R}\setminus \{\infty\}$ and variable $x$, $Q(a\otimes x)$ can be computed by scaling first row and second column of $Q(x)$ by $a$ and its multiplicative inverse (i.e., $-a$) respectively, where the format matrix $Q(f):= \begin{pmatrix}
f & 0\\
0 & \infty
\end{pmatrix}$ supports addition via $Q(f)\otimes Q(\infty)\otimes Q(g)= Q(f\oplus g)$. 


\begin{appsection}{Adapting Brent's Depth reduction for Formulas over $\mathsf{R}^{+}$}{Brent's depth reduction N appendix}

We are given a formula $F$ of size $s$ computing a polynomial over $\mathsf{R}^{+}$. Our goal is to build a formula of size $s^{\mathcal{O}(1)}$ and depth $\mathcal{O}(\log(s))$ that computes the same polynomial. This was already shown over rings by Brent \cite{brent1974parallel}. We make a minor change in his proof so that it also works over $\mathsf{R}^{+}$.

\textbf{This part is same as in Brent's proof}:\\
Start from the root node and at every step, pick that child of the current node for which the subformula rooted at it has larger size. Stop at the first node $v$ for which $\operatorname{size}(F_v)$ becomes $\leq \frac{2s}{3}$, where $F_v$ denotes the sub-formula of $F$ rooted at node $v$. We have $p(F) = \big(A \otimes p(F_v)\big) \oplus B$, where $A$ and $B$ are some polynomials over $\mathsf{R}^{+}$, and $p(F)$ and $p(F_v)$ denote the polynomials computed by formulas $F$ and $F_v$ respectively. Remove $F_v$ from $F$ and put a fresh variable (say $y$) in place of $v$, and let $F^{\text{new}}$ denote the formula so obtained. Note that $\operatorname{size}(F_v)\geq \frac{s}{3}$ and so, $\operatorname{size}(F^{\text{new}})\leq \frac{2s}{3}$. We have $p(F^{\text{new}}) = (A\otimes y) \oplus B$, where $p(F^{\text{new}})$ denotes the polynomial computed by the formula $F^{\text{new}}$. Note that $B = p(F^{\text{new}})\mid_{y=\infty}$ and $A\oplus B = p(F^{\text{new}})\mid_{y=0}$, where $p(F^{\text{new}})\mid_{y=\infty}$ and $p(F^{\text{new}})\mid_{y=0}$ are the polynomials obtained by substituting $y$ as $\infty$ and $0$ in the polynomial $p(F^{\text{new}})$ respectively.  

\textbf{This part is slightly different from Brent's proof}:\\
As per Brent's proof, the next step would say that $A = p(F^{new})\mid_{y=0} - B$. However, over semiring $\mathsf{R}^{+}$, we do not have the freedom to work with $-B$, i.e., additive inverse of $B$. We circumvent this minor difficulty as follows: Observe that $\big(B ~\otimes~  p(F_v)\big) \oplus B$ is same as $B$. That is, $\min\{B+p(F_v), B\} = B$. This is true because the polynomial $p(F_v)$ always evaluates to a non-negative value for any substitution of its variables from $\mathsf{R}^{+}$. So, $p(F) = \big(A \otimes p(F_v)\big)\oplus B$ can alternatively be written as $\big((A\oplus B) \otimes p(F_v)\big)\oplus B$. Next, plugging in expressions for $A\oplus B$ and $B$ obtained earlier, we get $p(F) = \Big(\big(p(F^{\text{new}})\mid_{y=0}\big)\otimes p(F_v)\Big)\oplus \Big(p(F^{\text{new}})\mid_{y=\infty}\Big)$. This expression gives a natural way to build a depth-reduced formula equivalent to $F$ (i.e., computing $p(F)$) given depth-reduced formulas obtained by recursion on $\leq \frac{2s}{3}$-sized formulas $F_v$ and $F^{\text{new}}$. The corresponding size and depth recurrences are $\operatorname{size}(s)\leq 3\cdot \operatorname{size}(2s/3) + \mathcal{O}(1)$ and $\operatorname{depth}(s)\leq \operatorname{depth}(2s/3)+\mathcal{O}(1)$. Solving these, we get $s^{\mathcal{O}(1)}$ and $\mathcal{O}(\log(s))$ as size and depth of the equivalent formula so built respectively.
\end{appsection}

\begin{appsection}{Analysis of Case 2 in the proof of Theorem \ref{non-universal width-2 semiring N}}{Case 2 non-universality N appendix}

\textbf{Case 2:} $j=i$\\
First, we show that the edge of $P_1$ from layer $i$ to layer $i+1$ must have its head at the bottom level of layer $i+1$. Suppose not. Then, $P_2$'s portion from layers $\leq i+1$ has one $x_2$, and $P_1$'s portion from layers $\geq i+1$ has one $y_1$. So, concatenating $P_2$'s portion from layers $\leq i+1$ with $P_1$'s portion from layers $\geq i+1$ gives us a source-to-sink path in $\Gamma$ whose weight is of the form $c\otimes x_2\otimes y_1$ for some $c\in \mathbb{R}$ (or $\mathbb{R}_{\geq 0}$). So, we have $(x_1\otimes y_1)\oplus (x_2\otimes y_2)\oplus (x_3\otimes y_3) \leq c\otimes x_2\otimes y_1$. Substituting $x_1=y_2=x_3=c+1$ and $y_1=x_2=y_3=0$ gives $c+1\leq c$, a contradiction. 

\begin{figure}[ht!]
\centering
\includegraphics[scale=0.8]{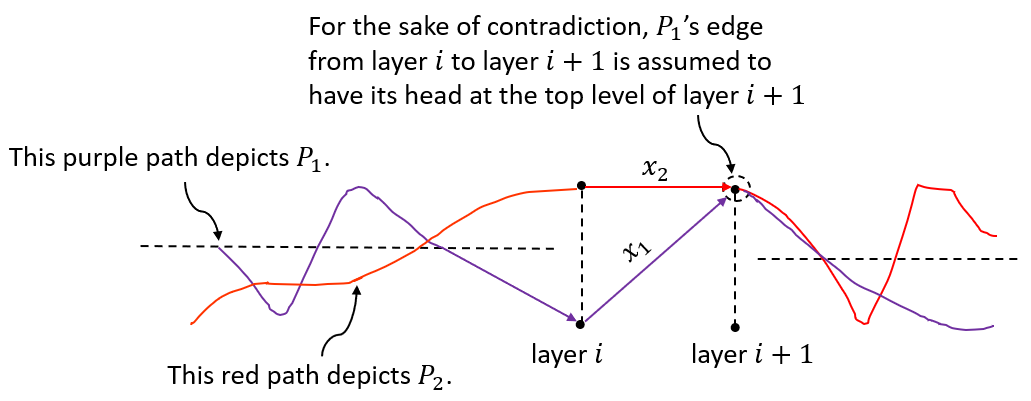}
\caption{}
\end{figure}

So, the edge of $P_1$ must have its head at bottom level of layer $i+1$. Also, using the same argument as in Case 1, the edge of $P_3$ must also have its head at bottom level of layer $j+1$ (same as $i+1$ here). Now, $P_3$'s portion from layers $\leq i+1$ has one $x_3$ or one $y_3$ or no variables. Also, $P_1$'s portion from layers $\geq i+1$ has one $y_1$. So, concatenating $P_3$'s portion from layers $\leq i+1$ with $P_1$'s portion from layers $\geq i+1$ gives us a source-to-sink path in $\Gamma$ whose weight is of the form $c\otimes y_1, c\otimes y_1\otimes x_3$ or $c\otimes y_1\otimes y_3$ for some $c\in \mathbb{R}$ (or $\mathbb{R}_{\geq 0}$). We derive a contradiction in these three cases as follows: In the first case, $(x_1\otimes y_1)\oplus (x_2\otimes y_2)\oplus (x_3\otimes y_3) \leq c\otimes y_1$. Substituting $x_1=x_2=x_3=c+1$ and $y_1=y_2=y_3=0$ gives $c+1\leq c$, a contradiction. In the second case, $(x_1\otimes y_1)\oplus (x_2\otimes y_2)\oplus (x_3\otimes y_3) \leq c\otimes y_1\otimes x_3$. Substituting $y_1=x_2=x_3=0$ and $x_1=y_2=y_3=c+1$ gives $c+1\leq c$, a contradiction. In the third case, $(x_1\otimes y_1)\oplus (x_2\otimes y_2)\oplus (x_3\otimes y_3) \leq c\otimes y_1\otimes y_3$. Substituting $y_1=y_2=y_3=0$ and $x_1=x_2=x_3=c+1$ gives $c+1\leq c$, a contradiction.

\begin{figure}[ht!]
\centering
\includegraphics[scale=0.8]{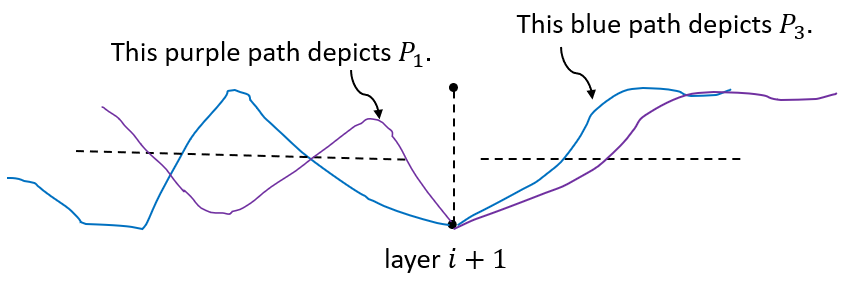}
\caption{}
\end{figure}
\end{appsection}

\begin{appsection}{Adapting Malod \& Portier's proof of $\mathsf{VNP} = \mathsf{VNF}$ over $\mathsf{R}$ and $\mathsf{R}^{+}$}{VNP=VNF adapted proof appendix}

This proof is almost same as the usual proof of $\mathsf{VNP}=\mathsf{VNF}$ over fields (see \cite{malod2008characterizing}). To show $\mathsf{VNP}_{\mathsf{R} (\text{ or } \mathsf{R}^{+})} \subseteq$ $\mathsf{VNF}_{\mathsf{R} (\text{ or } \mathsf{R}^{+})}$, it suffices to show $\mathsf{VP}_{\mathsf{R} (\text{ or } \mathsf{R}^{+})}\subseteq \mathsf{VNF}_{\mathsf{R} (\text{ or } \mathsf{R}^{+})}$. Consider any $p$-family  $(f_n)_{n\geq 1}\in \mathsf{VP}_{\mathsf{R} (\text{ or } \mathsf{R}^{+})}$. Then, $f_n(X)$ has a polynomial-sized circuit which, upon homogenization (see, for example, Lemma 5.2 in \cite{saptharishi2015survey}), can be converted into an equivalent polynomial-sized circuit of polynomially-bounded formal degree and then, can be converted into an equivalent polynomial-sized \emph{multiplicatively disjoint circuit} $C_n$ (i.e., for each of its $\bigotimes$ gates, the subcircuits rooted at both its children are vertex disjoint) using gate-cloning (see proof of Lemma 2 in \cite{malod2008characterizing}). 

A \emph{parse tree} $T$ of the multiplicatively-disjoint circuit $C_n$ is a subgraph of $C_n$ such that: i) The root node of $C_n$ belongs to $T$. ii) For every $\bigoplus$ gate of $C_n$ that belongs to $T$, exactly one edge of the two edges from its children belongs to $T$. iii) For every $\bigotimes$ gate of $C_n$ that belongs to $T$, both edges from its children belong to $T$. Also, the \emph{value} of $T$, denoted as $\operatorname{value}(T)$, is defined as the product of labels of leaf nodes of $C$ that belong to $T$. Note that $f_n(X) =\underset{T: ~T \text{ is a parse tree of } C_n}{\bigoplus   }\operatorname{value}(T)$. Any subgraph $H$ of $C_n$ can specified by the following $\infty$-$0$ assignment of indicator variables $p_{v}\mid_{v \in V(C_n)}$'s and $a_{(u,v)}\mid_{(u,v)\in E(C_n)}$'s: Set $p_v=0$ for every vertex $v\in V(C_n)$ that belongs to $H$, $a_{(u,v)}=0$ for every edge $(u,v)\in E(C_n)$ that belongs to $H$, and all other $p_v$'s \& $a_{(u,v)}$'s as $\infty$. Express $f_n(X)$ as follows:
\begin{equation*}
\begin{split}
f_n(X) &= \underset{\substack{\underline{p}~=~ p_v\mid_{v\in V(C_n)} \in \{\infty, 0\}^{|V(C_n)|}\\
\underline{a} ~= ~a_{(u,v)}\mid_{(u,v)\in E(C_n)}\in \{\infty,0\}^{|E(C_n)|}}}{\bigoplus   }\Big[\substack{\text{The subgraph corresponding to }\\(\underline{p}, \underline{a})\text{ is a parse tree of $C_n$ }}\Big]\otimes  \Big(\underset{\substack{u\in \operatorname{leaves}(C_n):\\p_u=0}}{\bigotimes}\operatorname{label}(u)\Big)\\
& = \underset{\underline{p}\in \{\infty,0\}^{|V(C_n)|}, ~\underline{a}\in \{\infty, 0\}^{|E(C_n)|}}{\bigoplus   } \Big[\substack{\text{The subgraph corresponding to }\\(\underline{p}, \underline{a})\text{ is a parse tree of $C_n$ }}\Big]\otimes  \Big(\underset{u\in \operatorname{leaves}(C_n)}{\bigotimes}\underbrace{\big(\operatorname{label}(u)\otimes  p_{u} ~\oplus    ~\overline{p_{u}} \big)}_{\text{Call } A_{u}. }\Big)
\end{split}
\end{equation*}
Further, we express the indicator $\Big[\substack{\text{The subgraph corresponding to }\\(\underline{p}, \underline{a})\text{ is a parse tree of $C_n$ }}\Big]$ as follows: 
\begin{equation*}
\begin{split}
&  \Bigg(\underset{(u,v)\in E(C_n)}{\bigotimes} \Big[(a_{(u,v)}=0)\Rightarrow (p_u=0~\& ~p_v=0)\Big]\Bigg)\otimes  [p_{\operatorname{root}}=0]\otimes   \bigg(\underset{\substack{u\in V(C_n):\\u \text{ is a $\bigotimes$ gate}}}{\bigotimes}\Big((p_u=0)\Rightarrow \big(\substack{a_{(v,u)}=0 \\ \text{ for both children $v$ of $u$}}\big)\Big)\bigg)\otimes  \\& \bigg(\underset{\substack{u\in V(C_n):\\u \text{ is a $\bigoplus$ gate}}}{\bigotimes}\Big((p_u=0)\Rightarrow \big(\substack{ a_{(v,u)}=0 \\ \text{ for exactly one } \text{ child $v$ of $u$}}\big)\Big)\bigg)\otimes  \bigg(\underset{u\in V(C_n)\setminus \{\operatorname{root}\}}{\bigotimes}\bigg[(p_u=0)\Rightarrow \big(\substack{\geq 1 \text{ out-neighbor $v$ of}\\ u \text{ is such that $p_v=0$}}\big)\bigg]\bigg)
\end{split}
\end{equation*}

\begin{equation*}
\begin{split}
&= \bigg(\underset{(u,v)\in E(C_n)}{\bigotimes}\underbrace{\big(\overline{a_{(u,v)}}\oplus    a_{(u,v)}\otimes  p_u\otimes  p_v\big)}_{\text{Call } B_{(u,v)}.}\bigg)\otimes  p_{\operatorname{root}}\otimes  \bigg(\underset{\substack{u\in V(C_n):\\u \text{ is a $\bigotimes$ gate}}}{\bigotimes}\underbrace{\Big(\overline{p_u} \oplus     ~~p_u\otimes  a_{(\ell(u),u)}\otimes  a_{(r(u),u)}\Big)}_{\text{Call } C_u.}\bigg)\otimes \\
&~~~\bigg(\underset{\substack{u\in V(C_n):\\u \text{ is a $\bigoplus$ gate}}}{\bigotimes}\underbrace{\Big(\overline{p_u}\oplus    p_u\otimes  \Big(a_{(\ell(u),u)}\otimes  \overline{a_{(r(u),u)}}\oplus     a_{(r(u),u)}\otimes  \overline{a_{(\ell(u),u)}}\Big)\bigg)}_{\text{Call } D_u.}\otimes \\
& \Bigg(\underset{u\in V(C_n)\setminus \{\operatorname{root}\}}{\bigotimes   }\underbrace{\big(\overline{p_u}\oplus    p_u\otimes  \underset{v:~(u,v)\in E(C_n)}{\bigoplus   }a_{(u,v)}\big)}_{\text{Call } E_u.}\Bigg),
\end{split}
\end{equation*}
where $\ell(u)$ and $r(u)$ denote the left child and right child of any node $u$ respectively.

Therefore, we have $f_n(X) = \underset{\substack{\underline{p}\in \{\infty,0\}^{|V(C_n)|}\\\underline{a}\in \{\infty,0\}^{|E(C_n)|}}}{\bigoplus   } g_n(X ,\underline{p}, \underline{a}, \overline{\underline{p}}, \overline{\underline{a}})$, where 
\begin{equation*}
g_n:= \Big(\underset{(u,v)\in E(C_n)}{\bigotimes}B_{(u,v)}\Big)\otimes  p_{\operatorname{root}}\otimes  \Big(\underset{\substack{u\in V(C_n):\\ u \text{ is a } \bigotimes \text{ gate}}}{\bigotimes}C_{u}\Big) \otimes  \Big(\underset{\substack{u\in V(C_n):\\ u \text{ is a } \bigoplus \text{ gate}}}{\bigotimes}D_{u}\Big) \otimes  \Big(\underset{u\in V(C_n)\setminus \{\operatorname{root}\}}{\bigotimes   }E_u\Big) \otimes  \Big(\underset{u\in \text{leaves}(C_n)}{\bigotimes}A_{u}\Big).
\end{equation*}
As $B_{(u,v)}$'s, $p_{\text{root}}$, $C_u$'s, $D_u$'s, $E_u$'s and $A_{u}$'s have polynomial-sized formulas, so does $g_n$. Thus, $(f_n)_{n\geq 1}\in \mathsf{VNF}_{\mathsf{R} (\text{or }\mathsf{R}^{+})}$.
\end{appsection}

\bibliography{references}

\newpage

\ifthenelse{\equal{\movetoappendix}{1}}{
       \appendix
       \section{Appendix}
       \includecollection{appendix}
} { }


\end{document}